\documentclass[11pt,letterpaper]{article}

\usepackage[margin=1in]{geometry}
\usepackage{xcolor}
\usepackage{latexsym,graphicx,amssymb}
\usepackage{amsmath}
\usepackage{amsthm, amsfonts}
\usepackage{mathtools, thmtools}
\usepackage{bbm}
\usepackage{float}
\usepackage{xspace}
\usepackage{paralist}
\usepackage{enumerate,multicol}
\usepackage{cases}
\usepackage{caption}
\usepackage{graphicx}
\usepackage{tikz,pgfmath}
\usepackage{pdfcomment}
\usetikzlibrary{matrix,positioning,quotes}
\usetikzlibrary{shapes.multipart}
\usetikzlibrary{calc}
\usetikzlibrary{arrows,decorations.markings}
\usetikzlibrary{matrix}
\usepackage{url}
\usepackage[ruled,linesnumbered,vlined]{algorithm2e}
\usepackage{hyperref}
\hypersetup{hidelinks}
\usepackage[capitalise]{cleveref}
\crefname{claim}{Claim}{Claims}
\Crefname{claim}{Claim}{Claims}
\crefname{statement}{Statement}{Statements}

\usepackage[noend]{algpseudocode}
\usepackage{multirow}
\usepackage{graphicx}
\usepackage{subcaption}
\usepackage{varwidth}
\usepackage{wrapfig}
\usepackage{enumitem}
\usepackage{float}
\usepackage{bbm,bm}
\usepackage{tcolorbox}

\usepackage[d]{esvect}
\usepackage{bm}
\usepackage{mleftright}
\usepackage{ulem}
\usepackage{stackengine}
\usepackage{mathrsfs}
\usepackage{bbm}

\usepackage[numbers]{natbib}

\newenvironment{proofof}[1]{\begin{proof}[Proof of #1]}{\end{proof}}
\numberwithin{figure}{section}

\newtheorem{theorem}{Theorem}[section]
\newtheorem{definition}[theorem]{Definition}
\newtheorem{problem}[theorem]{Problem}
\newtheorem{remark}[theorem]{Remark}

\newtheorem{lemma}[theorem]{Lemma}
\newtheorem{claim}[theorem]{Claim}
\newtheorem*{theorem*}{Theorem}
\newtheorem*{lemma*}{Lemma}

\input{head}
\renewcommand{\tilde}{\widetilde}
\renewcommand{\bar}{\widebar}
\renewcommand{\vec}[1]{\bm{\mathrm{#1}}}

\newcommand{\wordlen}{w}
\newcommand{\logstar}{\log^{*}}

\newcommand{\enkey}{n}  %
\newcommand{\enkeyi}[1][\ell]{\enkey_{#1}}

\newcommand\barbelow[1]{\stackunder[1.2pt]{$#1$}{\rule{1.1ex}{0.1ex}}}

\newcommand{\lbkey}{\barbelow{n}}
\newcommand{\lbkeyi}[1][\ell]{\lbkey_{#1}}
\newcommand{\ubkey}{\bar{n}}
\newcommand{\ubkeyi}[1][\ell]{\ubkey_{#1}}

\newcommand{\nbranchi}[1][\ell]{B_{#1}}

\newcommand{\dev}{\Delta}
\newcommand{\devi}[1][\ell]{\dev_{#1}}

\newcommand{\xid}{x_{\textup{id}}}
\newcommand{\xquot}{x_{\textup{quot}}}
\newcommand{\yid}{y_{\textup{id}}}
\newcommand{\yquot}{y_{\textup{quot}}}

\newcommand{\xrem}{x_{\textup{rem}}}
\newcommand{\idlen}{I}
\newcommand{\idleni}[1][\ell]{\idlen_{#1}}
\newcommand{\quotuniv}{Q}
\newcommand{\quotunivi}[1][\ell]{\quotuniv_{#1}}
\newcommand{\concat}{\mathop{\|}}
\newcommand{\permUnite}{h}
\newcommand{\chid}{v}
\newcommand{\Sid}{S_{\textup{id}}}

\newcommand{\hashfam}{\mathcal{H}}
\newcommand{\hashfami}[1][\ell]{\hashfam_{#1}}
\newcommand{\chidfami}[1][\ell]{\hashfam^{\textup{ch}}_{#1}}

\newcommand{\idfami}[1][\ell]{\hashfam^{\textup{id}}_{#1}}
\newcommand{\permCh}{\permUnite^{\textup{ch}}}
\newcommand{\permId}[1][v]{\permUnite^{\textup{id}}_{#1}}
\newcommand{\xprim}{x_{\textup{prim}}}
\newcommand{\yprim}{y_{\textup{prim}}}

\let\permchu\permChu
\let\permch\permCh
\newcommand{\permIdu}{\permId[u]}

\newcommand{\funcfam}{\mathcal{F}}  %

\newcommand{\ind}[1]{\mathbbm{1}\Bk{{#1}}}

\author{
  Tianxiao Li\thanks{Institute for Interdisciplinary Information Sciences, Tsinghua University. \url{litx20@mails.tsinghua.edu.cn}.}
  \and
  Jingxun Liang\thanks{Institute for Interdisciplinary Information Sciences, Tsinghua University. \url{liangjx20@mails.tsinghua.edu.cn}.} 
  \and
  Huacheng Yu\thanks{Department of Computer Science, Princeton University. \url{yuhch123@gmail.com}.}
  \and
  Renfei Zhou\thanks{Institute for Interdisciplinary Information Sciences, Tsinghua University. \url{zhourf20@mails.tsinghua.edu.cn}.}
}
\title{Dynamic Dictionary with Subconstant Wasted Bits per Key}
\date{}

\begin{document}
\maketitle

\thispagestyle{empty}
\setcounter{page}{0}

\begin{abstract}
  Dictionaries have been one of the central questions in data structures.
  A dictionary data structure maintains a set of key-value pairs under insertions and deletions such that given a query key, the data structure efficiently returns its value.
  The state-of-the-art dictionaries~\cite{liu22} store $n$ key-value pairs with only $O(n \log^{(k)} n)$ bits of redundancy, and support all operations in $O(k)$ time, for $k \leq \log^* n$.
  It was recently shown to be optimal~\cite{li2023lb}.

  In this paper, we study the regime where the number of redundant bits is $R = o(n)$, and show that when $R$ is at least $n/\text{poly} \log n$, all operations can be supported in $O(\log^* n + \log (n/R))$ time, matching the lower bound in this regime~\cite{li2023lb}.
  We present two data structures based on which range $R$ is in.
  The data structure for $R < n/\log^{0.1} n$ utilizes a generalization of \emph{adapters} studied in~\cite{BK0TW22,li2023dynamic}.
  The data structure for $R \geq n/\log^{0.1} n$ is based on recursively hashing into buckets with logarithmic sizes.
\end{abstract}

\newpage

\section{Introduction}
\label{sec:intro}

Dictionaries are among the most basic data structures. 
A dictionary maintains a set of key-value pairs, where the keys are distinct numbers from the key space $[U]$, and values are from the value space $[V]$.
The data structure supports key-value pair insertions and deletions such that given a key $x\in[U]$, it is able to quickly return the value associated to $x$, or report that $x$ does not exist in the current set.

Hash tables are textbook implementations of dictionaries.
Universal hashing together with chaining gives linear space and constant expected time per operation.
Extensive research has been conducted on dictionaries and hash tables~\cite{Knuth73, ANS09, ANS10, bender2021all, DdHPP06, PerfectHashing88, RealTime90, FPSS03, Storing84, knuth1963notes, LYY20, Pagh01, PAGH2004122, patrascu2008succincter, RR03, Yu20, liu22}.
Modern hash tables occupy space close to the information theoretical limit,
\[
  \mathcal{B}(U,V,n)=\log\binom{U}{n}+n\log V,
\]
to store $n$ key-value pairs.
A dictionary data structure that uses $\mathcal{B}(U,V,n)+rn$ bits of space is said to incur $r$ \emph{wasted bits per key}, and $R=rn$ is called the \emph{redundancy}.
The state-of-the-art hash tables by Bender, Farach-Colton, Kuszmaul, Kuszmaul and Liu~\cite{liu22} achieve $O(\log^{(k)} n)$ wasted bits per key\footnote{$\log^{(k)}n$ is the iterated logarithm.} while supporting all operations in $O(k)$ time, for any $k\leq \log^* n$.\footnote{Moreover, the hash table of~\cite{liu22} has the claimed operational time in worst case with high probability (high-probability), and occupies space in terms of the \emph{current} $n$ (dynamic-resizing). We will not focus on these stronger properties in this paper.}
Interestingly, this was later shown to be optimal~\cite{li2023lb} in the cell-probe model.
In particular, when a constant number of wasted bits per key is allowed, the optimal operational time is $\Theta(\log^* n)$.

On the other hand, the design of~\cite{liu22} ``wastes'' a constant number of bits per key in several components, making such a linear redundancy inherent to their hash tables.
It was not known how to construct a nontrivial dictionary with $o(n)$ bits of redundancy and fast operational time.
In this paper, we study the space-time tradeoff when only a $o(1)$ bits are wasted per key: can we achieve a low operational time when $r=o(1)$, ideally close to $O(\log^* n)$, or do the dictionaries intrinsically need to slow down significantly when $r$ is sub-constant?
The lower bound of~\cite{li2023lb} shows that when the redundancy $R$ is $o(n)$, the operational time is at least $\Omega(\log^* n+\log(n/R))$, which does not rule out $O(\log^* n)$-time when the redundancy is slightly sublinear.

The main result of this paper shows that this tradeoff is tight when $R$ is not too small.
\begin{theorem}
  \label{thm:main}
  For $R>n/\log^{O(1)}n$, there is a dictionary data structure with redundancy $R$ that supports all operations in $O(\log^* n + \log (n/R))$ expected time, in word RAM with word-size $w=\Omega(\log UV)$. Moreover, queries take constant expected time.
\end{theorem}

In fact, we use two different constructions to cover the range of $R$ between $n/\poly\log n$ and $n$.
The first data structure has $O(\log \log n)$ operational time when $R > n/\log^{O(1)} n$. It is based on the recently introduced \emph{adapters}~\cite{li2023dynamic} and dynamic aB-trees (daB-trees).
Adapters connect multiple variable-length data structures efficiently with \emph{zero} extra redundancy, overcoming an essential obstacle in designing succinct dynamic data structures.
Li, Liang, Yu and Zhou~\cite{li2023dynamic} presented an adapter connecting two variable-length data structures (two-way adapters).
In this paper, we extend adapters to $B$-way for $B=O(\sqrt{\log n})$, and apply it recursively to obtain daB-trees with large branching factors for dictionary, showing another application of the adapters.
The second data structure matches the lower bound when $R>n/\log^{0.1} n$.
It is based on recursively hashing into buckets that are logarithmically smaller in size, hence, constructing a tree of depth $O(\log^* n)$.
For leaves (the base cases), we introduce \emph{minimaps}.
They are ``mini-dictionaries'' maintaining a small set of keys, which can be encoded in one word, together with their (potentially large) values, and achieving almost no redundancy and all operations in logarithmic time.
Interestingly, $B$-way adapters and minimaps share similar constructions, while they play completely different roles in the data structures for the two regimes respectively.

\smallskip

We remark that our dictionary for $R > n / \log^{O(1)} n$ with $O(\log \log n)$ operational time is \emph{strongly history-independent} (a.k.a.\ \emph{uniquely representable}) if carefully implemented, which means that the memory state of the dictionary only depends on the set of stored key-value pairs as well as some random bits that the algorithm might use. This matches a lower bound in \cite{li2023lb}: any strongly history-independent dictionary must have operational time $\Omega\bigbk{\log \frac{n \log U}{R}}$, for any $R \le O(n \log U)$. Recently, Kuszmaul~\cite{Kus23Strongly} shows an optimal strongly history-independent dictionary with $\Theta\bigbk{\log \frac{n \log U}{R}}$ operational time for $R \in [n \log \log n, \, n \log n]$, which together with our result concludes the following theorem.

\begin{theorem}
  \label{thm:hist_ind}
  For $U = n^{1 + \Theta(1)}$ and $n / \log^{O(1)} n \le R \le n \log n$, there is a strongly history-independent dynamic dictionary that stores $n$ elements from universe $[U]$ with $R$ bits of redundancy and $O\bigbk{\log \frac{n \log U}{R}}$ expected time per operation, which is optimal in the cell-probe (or word RAM) model with word-size $w = \Theta(\log U)$. The dictionary also achieves expected constant-time queries.
\end{theorem}

\subsection{Organization}
In \cref{sec:overview}, we give an overview of the two data structures.
In \cref{sec:adapter}, we formalize the setup and present the constructions of $B$-way adapters and minimaps.
Then the data structure via $B$-way adapter trees is presented in \cref{sec:adaptertree}, and the data structure via multiple levels of hashing is presented in \cref{sec:multihash}.

\section{Overview}
\label{sec:overview}

In this section, we focus on the case where $U,V,n$ are bounded by polynomials of each other, and the word-size $w=\Theta(\log n)$.
We start by giving a short summary of the hash table by Bender et al.~\cite{liu22}.
The first step, similar to several prior works, is to hash into buckets that contain expected $K-K^{2/3}$ keys for some $K=\poly\log n$.
Then by concentration, all buckets have at most $K$ keys hashed to it with high probability.
By allocating space that fits $K$ key-value pairs to each bucket and handling them separately, this effectively reduces the number of keys to $K$, and have incurred wasted bits of roughly $1/K^{1/3}$ per key (from allocating space conservatively to each bucket).

Now within each bucket, we start to recursively hash the keys into smaller buckets, forming a tree structure where the buckets are the nodes.
For a bucket of size $S$ (i.e., that can store a total of $S$ keys), its children in the tree will have size $\poly\log S$ each, i.e., the degree of this node is $S/\poly\log S$.
The entire tree has $k$ levels for some parameter $k$, hence, every leaf in the tree has size roughly $\log^{(k)} n$.

All keys (and their values) are only directly stored in the leaves.
The entire tree rooted at a top-level bucket of size $K$ has exactly $K$ total slots to store the keys.
If the hash function were able to completely evenly split the key set, then a key and its value can simply be stored in the leaf it hashes to, and can be located by looking up this leaf.
However, a node $u$ may hash more keys to some child $v$ than others, resulting in $v$ having more keys than its capacity, i.e., $v$ overflows.
In this case, some of the keys will be stored in the empty slots in $v$'s siblings.
Provided that $u$ does not overflow, we can always find slots for all keys hashed to $u$.
An extra pointer will be stored to locate a key that is stored in a different child, incurring a redundancy of $\log S$ bits when $u$ has size $S$.
A simple calculation shows that by hashing $S$ keys into $S/\poly\log S$ buckets, the total number of keys that are ``overflown'' to a different bucket is $S/\poly\log S$.
Hence, the pointers incur $1/\poly\log S$ wasted bits per key on average.
Finally, the leaf nodes are sufficiently small, and data structures with $O(\log^{(k+1)} n)$ wasted bits per key are used, e.g., by storing a pointer for each key.
It turns out that the keys can be inserted or deleted by walking down the tree, taking $O(k)$ time.

\subsection{Dictionaries for \texorpdfstring{$R>n/\log^{0.1}n$}{R > n / log\^{}\{0.1\} n}}

For $R>n/\log^{0.1}n$, one contribution of this work is to handle the leaves efficiently while incurring only $o(1)$ wasted bits per key.
The leaves contribute a majority of the redundancy in~\cite{liu22}.
Moreover, their data structure always needs to store a ``query router'' in each leaf, which points each key to a slot.
The query router costs an average of $\Omega(1)$ wasted bits per key to distinguish all keys in addition to the pointers, even when $k$ is close to $\log^* n$.
This is unaffordable when only $o(1)$ wasted bits per key is allowed.

In this work, we introduce \emph{minimaps}, and use them to store the key-value pairs in the leaves.
Roughly speaking, a minimap maintains $n$ key-value pairs with keys from $[U]$ and values from $[V]$, such that \emph{the set of keys fits in one word}, i.e., $\log\binom{U}{n}\leq w$ (the values of $U,V,n$ here may be different from those in the full dictionary data structure).
Note that this task is nontrivial since $U$ and $n$ may be close but both super constant, and $V$ can still be very large, e.g., one could think $n=\sqrt{w}$, $U=w^{O(1)}$ and $V=2^w$.
Hence, the main challenge here is to map the keys to $n$ slots to store their \emph{values}, while supporting efficient key insertions and deletions.\footnote{The values can be stored using the \emph{spillover representation}~\cite{patrascu2008succincter} to avoid the problem of rounding $V$ to the next power of two.}

To rephrase the task, the set of keys is cheap enough to encode in one word, we want to store their values in $n$ slots with (almost) no redundancy such that given a key, we know where its value is stored by reading the entire key set, and when updating the key set by inserting or deleting one key, not many values need to be relocated.
A trivial solution is to store the values in the increasing order of their keys.
Since the entire key set can be achieved in constant time, given a key, we can compute its rank within the set, and hence, recover the index of the slot containing its value.
However, this solution may take $O(n)$ time to update the key set, since we may need to shift all values by one slot when inserting or deleting a key with a small rank.
Another solution is to keep the values in an arbitrary order while storing a pointer for each key to its value.
This achieves constant-time update and query, but uses $O(n\log n)$ redundancy bits in total (this is what the prior work uses).
Our (randomized) minimap achieves $O(\log n)$ expected time for insertion and deletion of a key-value pair, as well as constant lookup time.
Moreover, it does \emph{not} introduce \emph{any} redundancy.

The design of a minimap is inspired by the consistent hashing~\cite{karger1997consistent}, and has a very similar structure to a solution to \emph{memoryless worker-task assigning}~\cite{BK0TW22}, a problem studied in distributed computing.
Roughly speaking, we first map $n$ keys and $n$ slots to the unit circle randomly, and if a key $x$ is right after a slot $a$, we store the value of $x$ in $a$.
Then we repeat this process on the remaining keys and slots using independent randomness.
Each round matches a constant fraction of the keys and slots, hence, all keys are matched to slots after $O(\log n)$ rounds with high probability.
One can show that after one update, at most one key in each round needs to be relocated to a different slot.
Therefore, the update time is $O(\log n)$.
Besides the random functions, which can be shared between all minimaps, this structure has no redundancy, since the locations of all values are entirely determined by the key set (and the hash functions).

Using minimaps, we construct a dictionary with sublinear redundancy.
The main structure is similar to~\cite{liu22}, which recursively hashes keys to logarithmically smaller buckets.
Our construction also forms the same tree where each node corresponds to a bucket.
We note that there are other technical changes in order to keep the redundancy $o(n)$.
One such challenge comes from storing the keys: the fact that a key is hashed to a node carries information about the key, i.e., naively encoding the key would implicitly waste space.
Hence, when encoding a key, we must \emph{exploit} the fact that it has been hashed to this bucket (i.e., this fact shrinks the set of possible keys).
A common solution to this is the technical of \emph{quotienting}~\cite{Knuth73,Pagh01}.
Extra care needs to be taken when there is multiple levels of hashing as in our structure.

To avoid issues that may arise with variable-size data structures, we always allocate a fixed amount of space to each node.
Now not all keys are stored in the leaves, each internal node is also allocated a fixed number of slots.
We calculate the \emph{minimum} number of keys that will be hashed to a leaf with high probability, allocate and store exactly this number of keys in this leaf, and send the excessive keys to its parent node.
For each internal node $u$ in the tree, we also calculate the minimum number $m(u)$ of keys that will be hashed to $u$, and allocate a total of exactly $m(u)$ keys to $u$ and all its descendants, i.e., the number of keys stored directly in $u$ is determined by subtracting the total number of slots allocated to its children.
When a child of $u$ violates this minimum number, we will rehash the \emph{subtree} rooted at $u$, which only happens with low probability (in terms of the size of $u$).
Note that this strategy incurs another implicit redundancy by allowing a key to be stored in one of the many levels, i.e., for a given key, whether it is in the leaf or is stored directly in an ancestor is some information not determined by the input but implicitly determined by the data structure.
However, a careful calculation shows that this redundancy is also negligible.
By following this construction to the level such that each leaf roughly stores $O(n/R)$ keys, we achieve the claimed performance.

\subsection{Dictionaries for \texorpdfstring{$R<n/\log^{0.1}n$}{R < n / log\^{}\{0.1\} n}}
The above strategy fails when the redundancy is $O(n/\log^{10} n)$.
The primary reason is that this redundancy would require each leaf node to store (a large-)$\poly\log n$ keys while incurring only a constant bits of redundancy.
However, the minimap requires that the key set fits in one word, which does not hold when there are much more than $w=\log n$ keys.

After the first step of hashing into buckets of size $K=\poly\log n$, in this regime, we construct directly a data structure for each bucket with $O(1)$ bits of redundancy and supporting updates and queries in $O(\log\log n)$ time, based on the recently introduced \emph{adapters}~\cite{li2023dynamic}.

An adapter stores jointly $B$ arrays $D_1,\ldots,D_B$ using exactly $\sum_{i=1}^B\left|D_i\right|$ entries (with no redundancy), such that one can access one entry of an array $D_i$, or efficiently resize an array by incrementing or decrementing its size (and hence also the total number of entries in all arrays) by one.
The work of~\cite{li2023dynamic} designed a deterministic adapter for $B=2$ (two-way adapters) such that the entries can be accessed in constant time, and each resizing operation can be performed in $O(\log (\left|D_1\right|+\left|D_2\right|))$ time.
Adapters are useful to combine multiple components of a data structure, especially when each component has a variable length under updates.
Each array $D_i$ corresponds to the memory words of a data structure, and without incurring any redundancy, an adapter ``connects'' multiple data structures whose sizes may change.

In this paper, we generalize it to $B$-way adapters for larger $B$ using the essentially same algorithm as \emph{minimaps} described in the previous subsection, supporting accesses in constant time and resizing in $O(\log \sum \left|D_i\right|)$ time \emph{in expectation} (now the construction is randomized).
Roughly speaking, an adapter is to map the entries $\{(i, j):i\in [B], j\in[\left|D_i\right|]\}]$ to the indices $\bigBk{\sum_{i=1}^B\left|D_i\right|}$ with no redundancy, i.e., the mapping only depends on the array sizes, such that when some $\left|D_i\right|$ changes by one, only very few entries are mapped to a different location in expectation.

The seminal work of P\v{a}tra\c{s}cu~\cite{patrascu2008succincter} showed that augmented $B$-trees (aB-trees), a tree data structure with branching factor $B$ where each node stores some auxiliary information on its label that only depends the labels of its children, can be succinctly encoded with only $O(1)$ bits of redundancy, and can be accessed with no overhead in time (for $B$ not too large).
In particular, it can be used to build a \emph{static} dictionary for \emph{each bucket} of size $K=\poly\log n$ with $O(1)$ redundancy and lookup time $O(\log_B K)=O(\log_B\log n)$.
Two-way adapters were used to dynamize the augmented $B$-trees of P\v{a}tra\c{s}cu~\cite{patrascu2008succincter} for $B=2$ in~\cite{li2023dynamic}.
Roughly speaking, the succinct representation of a static aB-tree is constructed bottom-up: Having constructed the data structures for the children of a node $u$, the data structure for $u$ is obtained by first concatenating the data structures for the children of $u$, then encoding their labels succinctly together with the concatenation.
The adapters are used to store them succinctly with no space overhead, and allow for fast resizing.
It turns out that dynamizing aB-trees~\cite{li2023dynamic} results in another polynomial slowdown in the depth of the tree (in addition to the cost of adapters).
For each bucket of size $\poly\log n$, a tree with branching factor $B=2$ has depth $\log\log n$, and would be slowed down by a $\poly\log\log n$ factor.

Now we have efficient $B$-way adapters, and they give the potential to dynamize augmented $B$-trees for larger $B$.
In fact, we can set $B$ to $\sqrt{w}=O(\sqrt{\log n})$ (the largest $B$ that even the static aB-trees allow), and have a tree with constant depth.
It turns out that $B$-way adapters alone are not enough to dynamize general aB-trees for large $B$.
Since adapters only allow for incrementing or decrementing each data structure by one complete word (exactly $w$ bits), and the last less-than-$w$ bits of every component (an incomplete word) need to be handled separately also with no extra redundancy,
it is not clear how to dynamically maintain $B$ arrays that have incomplete words of arbitrary number of bits in $[0,w)$ that may resize.

One important fact that we crucially exploit in this work is that the aB-tree implementing dictionary \emph{does not resize arbitrarily}, as the insertion and deletion of a key will always resize the data structure by $\log (U/n)\pm O(1)$ bits.
By setting the entries size in the $B$-way adapter to be $\log (U/n)$, we always only need to resize by $O(1)$ bits in addition to one complete word.
This allows us to maintain the incomplete words efficiently using the same strategy as the adapters.

The question of dynamizing general aB-trees for large $B$, where the incomplete words may resize arbitrarily, remains open.
Resolving this question may improve the running times of several dynamic succinct data structures studied in~\cite{li2023dynamic}.

\section{Adapters}
\label{sec:adapter}

\newcommand{\matching}[1][]{\sigma_{\smallsub #1}}
\newcommand{\dist}{\Delta}
\newcommand{\udef}{\bot}
\newcommand{\incword}{\tilde{m}}
\newcommand{\numturn}[1][]{T}

\newcommand{\supp}{\mathop{\textup{supp}}}
\newcommand{\vl}{\vec \l}
\newcommand{\adapter}[1][]{\sigma_{\smallsub #1}}
\newcommand{\br}{B}
\newcommand{\Lmax}{L_{\max}}

It is a fundamental problem to manage multiple variable-size data structures within a contiguous piece of memory. The main purpose of this section is to introduce \emph{adapters} to help resolve this problem. The most important subroutine of adapters is to maintain an \emph{address mapping} between multiple small pieces of memory and a large one, which is further abstracted as the \emph{dynamic matching problem} which we will start with. Dynamic matching has various applications, including but not limited to the adapters.
We note that the solution here uses the same idea as a memoryless worker-task assigning algorithm~\cite{BK0TW22}.

\subsection{Dynamic Matching}
\label{sec:dynamic_matching}

\begin{problem}[Dynamic matching]
  \label{prob:dyn_matching}
  The \emph{dynamic matching} problem asks us to dynamically maintain an \emph{injection} (a matching) $\matching$ from a set of balls $A \subseteq U$ to a set of bins $B \subseteq V$, under the insertions and deletions of balls and bins (the operations).
  It is promised that $|A| \le |B|$ at any time. 
  Furthermore, the matching $\matching$ is required to only depend on $A, B$ (and possibly the random bits fixed in advance), written as $\matching = \matching[A, B]$.

  The \emph{cost} of an operation is defined as the number of ball relocations.
  Suppose that after a single operation, the new ball set and bin set become $A'$ and $B'$ respectively, the cost is
  \[
    \dist(\matching[A, B], \matching[A', B']) \defeq \sum_{a \in A \cup A'} \ind{\matching[A, B](a) \ne \matching[A', B'](a)}.
  \]
\end{problem}

A (deterministic) solution to this problem is described by a collection of matchings $\midBK{\matching[A, B]}_{A \subseteq U, \, B \subseteq V, \, |A| \le |B|}$, i.e., by assigning a matching to every possible pairs of ball set $A$ and bin set $B$. We call it a \emph{matching scheme}.
A randomized matching scheme is a distribution over the deterministic ones.

The main goal of this subsection is to introduce a randomized matching scheme with expected cost $O(\log |A|)$, requiring only a small number of balls to be moved during each operation.
We will sample the whole matching scheme $\midBK{\matching[A, B]}_{A, B}$ from a fixed distribution $\mathcal{D}$; after it is sampled and fixed, the maintained matching is determined solely by $A$ and $B$, i.e., the behavior of ball-relocations during each operation is deterministic.

\begin{lemma}
  \label{lm:dyn_matching}
  There is a randomized matching scheme $\midBK{\matching[A, B]}_{A, B} \sim \mathcal{D}$ such that the expected cost of any single operation on $(A, B)$ is $O(\log |A|)$.
\end{lemma}

\paragraph*{Matching Scheme.}

To prove \cref{lm:dyn_matching}, we first introduce a matching algorithm that takes $(A, B)$ as input and computes the matching $\matching[A, B]$. Following that, we will analyze its expected cost on a single operation.

Our algorithm is inspired by the \emph{consistent hashing}~\cite{karger1997consistent}, but with multiple rounds.
In each round, we use a hash function $h_i$ to map all remaining balls and bins to random points on the unit circle (i.e., the interval $[0, 1]$ with its two endpoints connected to form a loop). Suppose $a \in A$ is a ball whose next point in the clockwise direction on the circle is a bin~$b \in B$, then we create a matched pair $(a, b)$, i.e., we put ball $a$ into bin $b$. Intuitively, we expect half of the balls to be matched in such a round, so the whole algorithm ends after $O(\log |A|)$ rounds with high probability. The formal description of this algorithm is shown in \cref{dyncode}.

\begin{algorithm}[H]
	\caption{Dynamic Matching Algorithm}
	\SetKwRepeat{Do}{do}{while}
	\label{dyncode}
	\DontPrintSemicolon
	
	Hash functions $h_1,\dots,h_{|U|}$ are predetermined, where $h_i$ maps $U\cup V$ to the unit circle\;
	Let $A_i, B_i$ denote the sets of unmatched balls and bins before the $i$-th round, respectively\;
	$A_1 \gets A$, $B_1 \gets B$, $i \gets 1$\;
	\While{$A_i \ne \emptyset$} {
		$A_{i + 1} \gets A_i$, $B_{i + 1} \gets B_i$\;
		Hash all elements in $A_i$ and $B_i$ onto the circle according to $h_i$\;
		\For{$a \in A_i$} {
			\If {the element next to $a$ in the clockwise direction on the circle is a bin $b \in B_i$} {
				Set $\matching[A, B](a) = b$\;
				Remove $a$, $b$ from $A_{i+1}$, $B_{i+1}$\;
			}
		}
		$i \gets i + 1$\;
	}\end{algorithm}

The hash functions $\midBK{h_i}_{i \ge 1}$ are predetermined and stored in the memory in advance, and remain fixed throughout all operations. Therefore, when we insert or delete an element in $A$ or $B$, the images of other elements under $h_i$ remain unchanged. 
For any fixed pair $(A, B)$, over the randomness of $h_i$, the order of $h_i(A\cup B)$ on the circle is uniformly at random. 

The algorithm consists of multiple rounds, and denote by $\numturn$ the number of rounds. A worst-case upper bound is $\numturn \le |A|$, because in each round, at least one ball $a \in A$ will be matched. This also shows why using $|U|$ independent hash functions is always enough. Furthermore, we can prove that the expected number of rounds is $O(\log |A|)$.

\begin{claim}
  \label{clm:expected_round}
  For any $A \subseteq U$, $B \subseteq V$ with $|A| \le |B|$, if $\{h_i\}$ are uniformly random permutations, we have $\E[\numturn] = O(\log |A|)$.
\end{claim}

\begin{proof}
  Consider the number of pairs that are matched in the $i$-th round. We first list all $|A_i| + |B_i|$ points on the unit circle in clockwise order, starting from some globally fixed point. Denote by $x_1, \ldots, x_{|A_i| + |B_i|}$ the elements on the list. Due to the randomness of the hash function $h_i$, the list follows uniform distribution over all $(|A_i| + |B_i|)!$ permutations of $A_i \cup B_i$. If for some $j \in [1, |A_i + B_i|]$, $x_j$ is a ball while $x_{j + 1}$ is a bin, then we found a matched pair ($x_{|A_i| + |B_i| + 1} \defeq x_1$ for simplicity). For each of these events, the probability of it occurring is
  \[
	\Pr\BigBk{x_j \textup{ is a ball} \land x_{j+1} \textup{ is a bin}} = \frac{|A_i|}{|A_i| + |B_i|} \cdot \frac{|B_i|}{|A_i| + |B_i| - 1} > \frac{|A_i| \cdot |B_i|}{(|A_i| + |B_i|)^2}.
  \]
  Multiplying by the number of events, the expected number of matched pairs we find in the $i$-th round is at least $\frac{|A_i| \cdot |B_i|}{(|A_i| + |B_i|)^2} \cdot (|A_i| + |B_i|) \ge \frac{|A_i|}{2}$. This implies $\E[|A_{i + 1}|] \le \frac{|A_i|}{2}$. Applying Markov's inequality on $|A_{i + 1}|$, we know that $\Pr\Bk{|A_{i + 1}| \le \frac{3}{4}|A_i|} \ge \frac{1}{3}$. As the number of balls decreases by a constant factor $3/4$ with constant probability $1/3$ in each round, the entire process ends within $O(\log |A|)$ rounds in expectation.
\end{proof}

\cref{dyncode} has induced a matching scheme $\midBK{\matching[A, B]}_{A, B}$. We also know that the expected rounds of \cref{dyncode} is $\E[\numturn] = O(\log |A|)$. To prove \cref{lm:dyn_matching}, it remains to show that the cost of each operation on $(A, B)$ is bounded by $O(\numturn)$.

We use the matching scheme induced by \cref{dyncode}. Assume $(A, B)$ is a pair of ball set and bin set, and $(A', B')$ differs from $(A, B)$ by a single element, i.e., $(A', B')$ is obtained by inserting or deleting a single element in $A$ or $B$. Let $T$ and $T'$ be the number of rounds when running \cref{dyncode} on $(A, B)$ and $(A', B')$, respectively. From \cref{clm:expected_round}, we know that $\E[T] = O(\log |A|)$.

We still use $A_i, B_i$ to denote the sets of remaining balls and bins just before the $i$-th round when running \cref{dyncode} on $(A, B)$. Similarly, let $A'_i, B'_i$ denote the corresponding sets when running on $(A', B')$. We will compare $(A_i, B_i)$ with $(A'_i, B'_i)$, and show the following claim:

\begin{claim}
  \label{clm:dyn_induction}
  For all $1 \le i \le \max(T, T')$, $(A'_i, B'_i)$ differs from $(A_i, B_i)$ by a single element. Moreover, only $O(T)$ elements have different matched counterparts when the matching algorithm is executed on $(A', B')$ compared to $(A, B)$.
\end{claim}

\begin{proof}
  When $i = 1$, it is true because $(A'_1, B'_1) = (A', B')$ is obtained by a single operation from $(A, B) = (A_1, B_1)$. Below, we assume the claim holds for a certain $i$, and prove it for $i + 1$.
  
  $(A'_i, B'_i)$ is obtained from $(A_i, B_i)$ by inserting or removing a ball or bin. Due to the symmetry between insertions and deletions, as well as between balls and bins, let us only consider inserting a bin to obtain $B'_i = B_i \cup \midBK{b^*}$ for now. Consider how the matched pairs in the $i$-th round will change when we add a new bin $b^*$. There are three cases depending on the previous element of $b^*$ on the circle, denoted by $x$. Specifically, $x$ is the element adjacent to $b^*$ in the \emph{counterclockwise} direction.
  \begin{enumerate}
  \item $x$ is a bin $x \in B_i$, so $b^*$ stays unmatched and $(A'_{i+1}, B'_{i+1}) = (A_{i+1}, \, B_{i+1} \cup \{b^*\})$.
  \item $x$ is a ball $x \in A_i$ that is unmatched in round $i$ before adding $b^*$. Now the algorithm will match $x$ with $b^*$, which implies $(A'_{i+1}, B'_{i+1}) = (A_{i+1}\setminus\{x\}, \, B_{i+1})$.
  \item $x$ is a ball $x \in A_i$ that is matched to some other $b \ne b^*$ in this round before inserting $b^*$. Now $x$ will be matched to $b^*$ while $b$ becomes unmatched. In this case, $(A'_{i+1}, B'_{i+1}) = (A_{i+1}, \, B_{i+1} \cup \{b\})$.
  \end{enumerate}
  In all three cases, only $O(1)$ elements changed their matching in the $i$-th round, so the total number of such elements throughout all rounds is $O(\max(T, T'))$. 
  Finally, observe that $T' \le T + 1$, because $(A'_{T+1}, B'_{T+1})$ differs from $(A_{T+1}, B_{T+1})$ by a single element and $A_{T+1} = \emptyset$, thus $\bigabs{A'_{T+1}} \le 1$; since each round matches at least one ball, we know the algorithm must end within $T + 1$ rounds when running on $(A', B')$.
\end{proof}

Finally, combining \cref{clm:dyn_induction,clm:expected_round} proves \cref{lm:dyn_matching}.

\begin{remark}
  In our application of dynamic matching in the later sections, we always have $\left|A\right|=\left|B\right|$, and will only insert or delete a ball and a bin simultaneously.
\end{remark}

\subsection{Virtual Memory Model}

Same as \cite{li2023dynamic}, we use a storage model called the \emph{virtual memory model} to capture the essential memory-accessing behavior of variable-size data structures. Similar to the word RAM model, the memory consists of an infinite sequence of $w$-bit words, which are labeled with positive integers $\midBK{1, 2, \ldots}$ as their \emph{addresses}. One should view the memory as a tape that starts from word 1 and extends to infinity.

At any given time, the variable-size data structure is allowed to use a prefix of the memory string. Formally, there is a positive integer $M$ indicating the number of available memory bits, and let $L \defeq \midfloor{M / w}$. Like in a word RAM with $M$ bits, the data structure can read or write one of the first $L$ words on the tape given its address $i \in [L]$, as well as the first $M - Lw = (M \bmod w)$ bits within the $(L+1)$-th word. The latter part is viewed as an \emph{incomplete word}, which we also refer to as the \emph{tail}; in contrast, the former part is called the \emph{complete words}. We may use ``accessing a word'' to refer to reading or writing it. Accessing an arbitrary word (including the tail) counts as one \emph{word-access}.

We also allow the data structure to change the memory size $M$ via \emph{allocations} and \emph{releases}. Due to technical reasons, we only allow two types of allocations (releases): increasing (decreasing) $M$ by 1 or $w$ bits, but not between. The memory size $M$ is stored by some outside entity and is always given to the algorithm when performing any operation.

The available part on the tape (i.e., the first $L$ words plus $(M \bmod w)$ bits) is called a \emph{virtual memory} (VM for short), which is responsible of storing variable-size data structures. Usually, we think allocations and releases are more expensive than word-accesses; similarly, word-accesses to VM are more expensive than arithmetic instructions, because there might be multi-level address translations between the VM and the physical memory. So we observe the following quantities to measure the time performance of a variable-size data structure stored in a VM:
\begin{enumerate}
\item the number of arithmetic instructions and lookup table queries it performs;
\item the number of word-accesses to the VM;
\item the number of allocations and releases (including both 1-bit and $w$-bit types).
\end{enumerate}

\paragraph{Variable-size data structures in spillover representation.}

The \emph{spillover representation}, first introduced in \cite{patrascu2008succincter}, aims to avoid the 1-bit redundancy from rounding up the data structure to an integer number of bits. It represents the data structure with a pair $(k, m) \in [K] \times \midBK{0, 1}^M$, where $k$ is called the \emph{spillover} and $K$ is called the the \emph{spill universe}.

This representation can be naturally combined with the VM model when the represented data structure is dynamic and has variable size: We directly store the $M$ memory bits into a virtual memory, supporting word-accesses, allocations and releases. Then, any update to the data structure (including changing its size) is mainly expressed as a sequence of VM operations. The spillover $k$, as well as $K, M$, is stored by some external entity, so we do not count the time taken to it when we analyze the stored data structure itself.

\subsection{Adapters}
\label{sec:adapter_sub}

Recall that we want to store $B$ variable-size data structures $D_1, \ldots, D_B$ in a contiguous piece of memory. For now, let us focus on the simplest case: each $D_i$ is encoded into $M_i = \l_i w$ memory bits, which are stored in a VM, without tail or spillover. In simpler terms, we need to ``concatenate'' $B$ smaller VMs (sub-VMs) with $\l_i$ complete words and no tail, then store them into a larger VM (super-VM) with $L \defeq \l_1 + \cdots + \l_B$ words. In this subsection, we introduce \emph{$B$-way adapters} to solve this simple case. Concatenating tails and spillovers involve more technical details and will be introduced later in \cref{sec:dabtree}.

A $B$-way adapter maintains an address mapping between $B$ sub-VMs and the super-VM. Formally, we use $(i, j)$ to indicate the $j$-th word in the $i$-th sub-VM, while the words in the super-VM are labeled with $\midBK{1, 2, \ldots, L}$. Our adapter maintains a dynamic matching $\matching$ between $\midBK{(i, j) : 1 \leq i \leq B, \, 1 \le j \le \l_i}$ and $[L]$, which only depends on $\vl \defeq (\l_1, \ldots, \l_\br)$, written as $\matching = \matching[\vl]$. Suppose $\matching(i, j) = t$, then we store the content of word $(i, j)$ into the $t$-th word in the super-VM. The metadata $\vl$ is stored outside and we assume free access to it. Using this design, storing the complete-word parts of $\br$ data structures incurs no redundancy.

The matching scheme we use is given in \cref{sec:dynamic_matching}: By viewing the words in sub-VMs $(i, j)$ as balls and the words in the super-VM as bins, it is clear that it is an instance of the dynamic matching \cref{prob:dyn_matching}. We maintain the matching via lookup tables, which requires that $|\!\supp \vl|$ is small enough. When we want to access the $j$-th word in the $i$-th sub-VM, we first query the lookup table to get $\matching[\vl](i, j) = t$ in constant time, then access the $t$-th word in the super-VM.

When any of the sub-VM requests an allocation (or release), the super-VM should also perform the same request.
Afterward, some of the words should be moved as the matching changes from $\matching[\vl]$ to $\matching[\vl']$.
The number of moved words equals
\[
  \dist\bk{\matching[\vl], \matching[\vl']} = \sum_{i, j} \ind{\matching[\vl](i, j) \ne \matching[\vl'](i, j)}.
\]
By \cref{lm:dyn_matching}, its expectation is bounded by $O(\log L)$. That is, an allocation (or release) of a sub-VM is transformed to an allocation (or release) of the super-VM with $O(\log L)$ additional word-accesses in the super-VM in expectation.

We summarize the above discussion with the following lemma:

\begin{lemma}
  \label{lm:adapter}
  Assuming $\vl \defeq (\l_1, \ldots, \l_\br)$ is stored outside and allows free access, $\br$ sub-VMs with $\l_1,\ldots,\l_\br$ words respectively and no tails can be maintained within a single super-VM with $L \defeq \sum_{i=1}^\br \l_i$ words, while
  \begin{itemize}
  \item a word-access on any sub-VM is simulated by a word-access on the super-VM;
  \item an allocation [resp$.$ release] on a sub-VM is simulated by an allocation [resp$.$ release] and $O(\log L)$ word-accesses in expectation on the super-VM.
  \end{itemize}
  Moreover, if $L \le \Lmax$ always holds, we can precompute lookup tables occupying $O(B\Lmax^{B+1})$ words in total, and complete the simulation within $O(1)$ computation cost for each access, and $O(\log L)$ expected computation cost for each allocation or release. The time of computing these lookup tables is $O((B+\Lmax)\Lmax^{B+1})$.
\end{lemma}

\begin{proof}
  Based on the discussion above, it remains to calculate the size of lookup tables. We precompute the following two lookup tables:
  \begin{itemize}
      \item Given $\vl$ and an address $(i,j)$ in sub-VM, output the corresponding address $\matching[\vl](i,j)$ in the super-VM.
      \item Given an allocation/release changing $\vl$ to $\vl'$, find out the list of all word-moves in the super-VM during this operation, where each word-move is represented by a pair $(t,t')$, meaning that the content of word $t$ should be moved to word $t'$.
  \end{itemize}
  The number of vectors $\vl$ is bounded by $\Lmax^{B}$; for every $\vl$, there are $O(B)$ possible $\vl'$ that can be reached by a single operation. Each entry of the second table is a list of at most $\Lmax$ pairs, so the total size is bounded by $B\Lmax^{B+1}$. It is also clear that the first table occupies only $O(\Lmax^{B+1})$ words of memory.
  
  These two tables can also be computed efficiently. To compute the first one, we only need to run \cref{dyncode} on every possible $\vl$, which takes $O(\Lmax^{B+2})$ time. To compute the second one, we only need to compare $\matching[\vl](\cdot, \cdot)$ with $\matching[\vl'](\cdot, \cdot)$ for every adjacent pair $(\vl, \vl')$, which takes $O(B\Lmax^{B+1})$ time.
\end{proof}

\subsection{Minimap for Dictionary Problem}

Another application of dynamic matching is maintaining key-value pairs with almost zero redundancy. Suppose we need to maintain $n$ key-value pairs, supporting insertions and deletions of key-value pairs, and querying the associated value of a given key. Differently from the general dictionary problem, here we additionally require that $\log \binom{U}{n} = O(w)$, i.e., the whole key set can be encoded into $O(1)$ words.
Note that the problem is non-trivial, as the value-universe $V$ can still be large. We introduce variable-size data structures called \emph{minimaps} for this regime, stated as follows.

\begin{lemma}[Minimap]
  \label{lm:minimap_lem}
  Suppose $r > 0$ is an integer and the word-length is $w = \Omega(n \log (Ur) + \log V)$. There is a data structure that stores $n$ key-value pairs from key-universe $[U]$ and value-universe $[V]$, which is encoded using $M$ memory bits (stored in a VM) and a spillover with universe $K \le 2r$, has at most $4/r$ bits of redundancy, $O(1)$ query time, and $O(\log n)$ expected insertion/deletion time. The data structure is allowed to access a lookup table of $O\bigbk{Un \binom{U}{n} \cdot (2rn)^n}$ words that can be shared between multiple instances of minimaps with the same parameters $U, V, n, r$.
\end{lemma}

Besides the dynamic matching technique, we will also use the following lemma from~\cite{patrascu2008succincter}, which formalizes the fundamental functionality of the spillover representation.

\begin{lemma}[{\cite[Lemma 3]{patrascu2008succincter}}]
  \label{lm:succincter_lem3}
  Given an arbitrary set $\mathcal{X}$. Fixing a parameter $r > 0$, we can represent any element in $\mathcal{X}$ by a pair $(k, m) \in [K] \times \{0,1\}^{M}$ with $K \le 2r$, incurring at most $2/r$ bits of redundancy.
  If $w = \Omega(\log |\mathcal{X}|)$, encoding and decoding only take $O(1)$ arithmetic operations.\footnote{The original lemma in \cite{patrascu2008succincter} required $r \le |\mathcal{X}|$; however, when $r > |\mathcal{X}|$, it is easy to encode an element in $\mathcal{X}$ with a spillover $k \in [|\mathcal{X}|]$ and no memory bits, which incurs no redundancy, so the lemma holds as well.}
\end{lemma}

The main idea of our data structure is to establish $n$ ``slots'' in the memory, each storing a value in $[V]$. The order we store these values is maintained as a dynamic matching, i.e., once key $x$ is matched with slot $i$, we store the value associated with $x$ into slot $i$. During an insertion or deletion, \cref{lm:dyn_matching} guarantees that only few values need to be relocated.

\begin{proofof}{\cref{lm:minimap_lem}}
  First, we apply \cref{lm:succincter_lem3} on the value-universe $[V]$ with parameter $r' = rn$, obtaining a spillover representation $[V] \to [K^*] \times \midBK{0, 1}^{M^*}$ with $K^* \le 2rn$. Then, every value in $[V]$ is represented by a pair $(k^*_i, m^*_i) \in [K^*] \times \midBK{0, 1}^{M^*}$. This step induces $\le 2n / (rn) = 2/r$ bits of redundancy.

  In the second step, we combine the key set and all $n$ spillovers $k^*_1, \ldots, k^*_n$. Formally, let $\vec{x} = (x_1, \ldots, x_n, k^*_1, \ldots, k^*_n)$, where $x_1 < x_2 < \cdots < x_n$ are the $n$ distinct keys to be stored, and $k^*_i \in [K^*]$ is the spillover of $x_i$'s value. Let $\mathcal{X}$ be the set of all possible $\vec{x}$'s, then $|\mathcal{X}| \le \binom{U}{n} \cdot (2rn)^{n}$. We apply \cref{lm:succincter_lem3} again on $\mathcal{X}$ with parameter $r$, obtaining a spillover representation with spill universe $K \le 2r$, incurring at most $2/r$ bits of redundancy.
  Also note that $\log\left|\mathcal{X}\right|\leq O(w)$, the memory bits from this step only occupy $O(1)$ words.

  The third step is to partition $M^* \cdot n$ consecutive memory bits into $n$ slots of $M^*$ bits each, labeled from 1 to $n$. We maintain a dynamic matching between the key set $\midBK{x_1, \ldots, x_n} \subset [U]$ and slots $\midBK{1, \ldots, n}$. Once some key $x_i$ is mapped to slot $j$, we store the memory bits $m^*_i$ of $x_i$'s value in the $j$-th slot.

  Finally, we \emph{append} the memory bits from the second step ($O(1)$ words) \emph{to} the $M^* \cdot n$ memory bits from the third step. The result is stored into a VM and forms the final memory bits of our data structure (minimap). The spillover of the minimap $k \in [K]$ is the spillover from the second step. Hence, the total redundancy does not exceed $4/r$ bits.

  During an insertion or deletion, according to \cref{lm:dyn_matching}, only $O(\log n)$ slots need to relocate their contents, which takes $O(\log n)$ word-accesses. The outcome of the second step can be fully re-encoded as it takes only $O(1)$ word-accesses.

  It remains to calculate the sizes of the lookup tables we need in this process:
  \begin{itemize}
  \item Given $\vec{x}$ and a key $x_i$, decode the spillover $k^*_i$ of its value, and find the address of the slot where $x_i$'s associated value is stored. It occupies $O\bigbk{\binom{U}{n} \cdot (2rn)^n \cdot U}$ words.
  \item Given $\vec{x}$ and a key $x$ to be inserted or deleted, find out the updated $\vec{x}$ and a sequence of slot-moves (which slot's content should be moved to which location) during this operation. It requires $O(n)$ words for every possible operation, thus occupies $O\bigbk{\binom{U}{n} \cdot (2rn)^n \cdot U \cdot n}$ words.
  \end{itemize}
  Summing up, the lookup table occupies $O\bigbk{Un \binom{U}{n} \cdot (2rn)^n}$ words, which concludes the proof.
\end{proofof}

\section{Dictionary via Adapter Tree}
\label{sec:dabtree}
\label{sec:adaptertree}
\newcommand{\nodewidth}[1]{n_{#1}}
\newcommand{\nodewidthlb}[1]{\barbelow{n}_{#1}}
\newcommand{\nodewidthub}[1]{\bar{n}_{#1}}
\newcommand{\keyunisize}[1]{U_{#1}}
\newcommand{\numkey}[1]{n({#1})}
\newcommand{\numkeyu}{n{(u)}}
\newcommand{\numkeyvi}[1][i]{n{(v_{#1})}}
\newcommand{\quotientuni}{Q}
\newcommand{\zid}{z_{\textsc{id}}}
\newcommand{\zquot}{z_{\text{quot}}}
\newcommand{\numinstance}[2]{\binom{\keyunisize{#1}}{\numkey{#2}}}
\newcommand{\spilluni}[2]{K\bk{\keyunisize{#1}, \numkey{#2}}}
\newcommand{\spillvi}[1][i]{k_{#1}}
\newcommand{\nummembits}[2]{M\bk{\keyunisize{#1}, \numkey{#2}}}
\newcommand{\completewordi}[1][i]{m_{\textup{word}}^{[#1]}}
\newcommand{\extrabitsi}[1][i]{m_{\textup{bit}}^{[#1]}}
\newcommand{\extrabits}{m_{\textup{bit}}}
\newcommand{\mcat}{m_{\textup{cat}}}
\newcommand{\Mcat}{M_{\textup{cat}}}
\newcommand{\Mmax}{M_{\textup{max}}}
\newcommand{\mfix}{m_{\textup{fix}}}
\newcommand{\Mfix}{M_{\textup{fix}}}
\newcommand{\mrem}{m_{\textup{rem}}}
\newcommand{\Mrem}{M_{\textup{rem}}}
\newcommand{\numkeyvec}{\vec{n}}
\newcommand{\X}{\mathcal{X}}
\newcommand{\ym}{y_{\smallsub M}}
\newcommand{\yk}{y_{\smallsub K}}
\newcommand{\kcat}{k_{\textup{cat}}}
\newtheorem{statement}{Statement}
\newcommand{\numword}{L_{\textup{word}}}
\newcommand{\numwordi}[1][i]{L^{[i]}_{\textup{word}}}
\newcommand{\numextbit}{L_{\textup{bit}}}
\newcommand{\numextbiti}[1][i]{L^{[i]}_{\textup{bit}}}

Recall that a dictionary data structure maintains a set of $N$ distinct keys from $[U]$ under key insertions and deletions, such that given a query $x\in [U]$, it returns if $x$ is in this set. In some scenarios, the dictionaries may also need to store a $O(\wordlen)$-bit string called the \emph{value} for each key. For most algorithms including ours, storing values is not the main challenge, so we will focus on storing keys.
Now fix an integer $100 \le h \le \log^{0.3} N$ and let $R = N / \log^h N$. In this section, we will design a dictionary with redundancy $O(R)$, which supports queries in $O(h^2)$ time, and updates in $O(h^3 \log \log N)$ time. When $h = \Theta(1)$, this result covers the parameter regime $R \le N / \log^{0.1} N$ of \cref{thm:main}.

Additionally, we assume $U = N^{1 + \alpha}$ for some constant $\alpha > 0$; also assume $\wordlen = \Omega(\log N)$. Without loss of generality, we assume the dictionary always contains $N$ or $N - 1$ keys, but not fewer.\footnote{This assumption is without loss of generality because we may insert $N$ special elements $e_1, \ldots, e_N$ to the universe $U$, each indicating that the $i$-th key in the dictionary does not exist. If the number of keys is smaller than $N-1$, we include a suffix of special elements in the dictionary to make the number of keys equal $N-1$ or $N$ at any time. The entropy increased by enlarging the universe is negligible: $1/\poly(N)$.}

\paragraph*{Algorithm Framework.}
Let $B = O(\sqrt{\log N})$ be a parameter, $\nodewidth{1} = O(B^{8h})$, and $\nodewidth{8h} = \nodewidth{1} / B^{8h-1}=O(B)$. Consider the following two-stage hashing: We first hash keys into $N/\nodewidth{1}$ buckets, which we call the \emph{level-1 nodes} or \emph{root nodes}, with the expected number of keys in each bucket being $\nodewidth{1}$.
For all keys that mapped to a level-1 node $u$, we further hash them into $\nodewidth{1} / \nodewidth{8h}$ buckets, which we call \emph{level-$8h$ nodes} or \emph{leaf nodes}, with the expected number of keys in each bucket being $\nodewidth{8h}$.

Next, we build a full $B$-ary tree between every root and its leaves. Formally, for $1 \le \l < 8h$, every level-$\l$ node has exactly $B$ children at level $\l + 1$. We now get a forest where each key is finally hashed into a leaf node (through two-stage hashing). For a level-$\l$ node $u$, we say a key $x$ is hashed into $u$ if $x$ is hashed into a descendant leaf of $u$; the expected number of keys hashed into $u$ is denoted by $\nodewidth{\l} \defeq B^{8h-\l} \nodewidth{1}$.

The design of our data structure is recursive: For a level-$\l$ node $u$, we construct a variable-size data structure maintaining all keys hashed into $u$, which is stored in a \emph{virtual memory} (VM). This is mainly done by ``concatenating'' the sub-VMs from its $B$ children on the next level. The leaf nodes are directly maintained by \emph{mini-maps} as the number of keys hashed into every leaf node is very small ($O(\sqrt{\log N})$). We call this recursive data structure \emph{adapter trees}.

\paragraph*{Tree Parameters.}
We specify some parameters here. For technical reason, we choose $B$ to be the unique power of two in $[\sqrt{\log N}, 2\sqrt{\log N})$, and $\nodewidth{1} \in [B^{8h}, 2B^{8h})$ be the unique number such that $N / \nodewidth{1}$ is a power of two (note that $\nodewidth{1}$ is \emph{not} necessarily an integer). Correspondingly, $\nodewidth{\l} = \nodewidth{1} / B^{\l-1}$ is also not necessarily an integer.

Let $\nodewidthub{1} = \bigceil{\nodewidth{1} + \nodewidth{1}^{2/3}}$ and $\nodewidthlb{1} = \bigfloor{\nodewidth{1} - \nodewidth{1}^{2/3}}$. By carefully choosing the hash functions (as we will describe in \cref{sec:hashfunc_for_adapter}), we can ensure that the number of keys $\numkeyu$ in each level-1 node $u$ is between $[\nodewidthlb{1}, \nodewidthub{1}]$ with high probability.
When designing our data structure, we always assume this to be true.
Once some root $u$ violates this condition, we say $u$ \emph{overflows} ($\numkeyu > \nodewidthub{1}$) or \emph{underflows} ($\numkeyu < \nodewidthlb{1}$), and reconstruct the whole dictionary immediately (using different hash functions). Since the probability is sufficiently small, the expected cost of reconstruction is negligible (this will be analyzed in \cref{sec:hashfunc_for_adapter} and 
\cref{sec:rehashing_for_multilevel}). Similarly, we define $\nodewidthub{8h}= \bigceil{\nodewidth{8h} + \nodewidth{8h}^{2/3}}$ and $\nodewidthlb{8h}  = \bigfloor{\nodewidth{8h} - \nodewidth{8h}^{2/3}}$, and always assume that the number of keys in each level-$8h$ node is between $[\nodewidthlb{8h}, \nodewidthub{8h}]$.

When $N$ is sufficiently large, we know $\nodewidthlb{8h} \ge \nodewidth{8h} / 2$ and $\nodewidthub{8h} \le 2 \cdot \nodewidth{8h}$. Thus, for any given level-$\l$ node $u$, we know $\numkeyu \in [\nodewidth{\l} / 2, \, 2 \nodewidth{\l}]$, because the keys hashed to $u$ is a union of those of $u$'s descendent leaves.

\subsection{Hashing and ID-Quotient Separation}
Imagine that we have $n$ keys that are independent random bit strings. With probability $\ge 1 - 1 / n^2$, the first $\ceil{4 \log n}$ bits are all distinct, thus sufficient to distinguish the keys from each other. We call the first $\ceil{4 \log n}$ bits the \emph{identifier} (``ID'' for short) of each key, and call the remaining bits the \emph{quotient}. An original key can now be regarded as a \emph{key-value pair}: the ID is enough to identify a key while the quotient is considered as the value associated with the ID.

A merit of separating IDs and quotients is that the IDs can be much shorter than the original keys, which enables us to use efficient data structures for them. Typically, storing associated values in a dictionary is much easier than storing the keys. By putting most of the bits in the quotient, we actually reduce the key universe a lot. For $\nodewidth{8h} = O(\log^{0.5} N)$ keys, it is possible to encode the whole ID set into a single word while supporting constant time operations based on lookup tables, so we can maintain the ID-quotient pairs using a \emph{minimap}.

In our algorithm, the keys are represented in ID-quotient pairs since they were hashed into level-1 nodes. Formally, a level-1 node $u$ is responsible for storing a subset of its universe $[2^{\idleni[1]}] \times [\quotientuni]$, where $\idleni[1] \defeq \ceil{4 \log \nodewidth{1}}$ is the length of the ID part, and $\quotientuni$ is the range of the quotient part.\footnote{We allow the quotient range not to be a power of two.} The hash function for this hashing step is a \emph{bijection} $h^{(1)} : [U] \to [N/\nodewidth{1}] \times [2^{\idleni[1]}] \times [\quotientuni]$, which maps a key $x \in [U]$ into a triple $(u, \xid, \xquot)$, with $u$ being the index of a level-1 node, $\xid$ and $\xquot$ being the ID and quotient of key $x$ in the root node $u$, respectively. (We round up $U$ to a multiple of $(N / \nodewidth{1}) \cdot 2^{\idleni[1]}$ so that $h^{(1)}$ can be a bijection, which only increases the entropy $\log \binom{U}{N}$ by a negligible value.)

The second hashing step further hashes each key from a level-1 node to level-$8h$ nodes. Benefiting from the fact that the keys are already represented as ID-quotient pairs, the second hashing only takes the ID $\xid$ as input.
Specifically, this hash function maps $\xid$ (the ID of $x$ in the root node $u$) into a pair $(v, \yid)$, where $v \in [B^{8h-1}]$ indicates a descendant leaf of the root $u$, and $\yid$ is the ID of $x$ in the leaf $v$. Formally, this hash function is a bijection $h^{(2)} : [2^{\idleni[1]}] \to [B^{8h - 1}] \times [2^{\idleni[8h]}]$, where $\idleni[8h] = \idleni[1] - (8h-1) \log B$. The quotient remains the same, i.e., the ID-quotient pairs for $x$ in the leaf node is $(\yid, \xquot)$.

The hash function $h^{(1)}$ should satisfy two properties: (1) it should evenly distribute all keys to the target buckets, preventing overflow or underflow at any bucket; (2) for all keys that are hashed to the same bucket, they must be assigned different IDs in it. $h^{(2)}$ also needs to satisfy the first property. We defer the formal design of hash function families to \cref{sec:hashfunc_for_adapter}.

After two hashing steps, each leaf is responsible for storing a set of ID-quotient pairs from $[2^{\idleni[8h]}] \times [\quotientuni]$. We use $\keyunisize{8h} \defeq 2^{\idleni[8h]} \cdot \quotientuni = U \cdot \nodewidth{8h} / N$ to denote the universe size of leaf nodes. Accordingly, we define $\keyunisize{\l} \defeq \keyunisize{8h} \cdot B^{8h - \l}$, representing that any level-$\l$ node $u$ is responsible for storing $\numkeyu$ keys from a universe of $\keyunisize{\l}$ elements. We assume $\keyunisize{\l} \ge 4 \nodewidth{\l}$ for all $1 \le \l \le 8h$ as $N$ is sufficiently large, which also implies $\keyunisize{\l} \ge 2 \numkeyu$ for any level-$\l$ node $u$.

\subsection{Weak Virtual Memory}
\label{sec:wvm}

Recall that on each non-leaf node $u$ of the tree, we plan to inductively maintain a variable-size data structure that stores the set of keys hashed to $u$, mainly by concatenating the sub-VMs from $u$'s children. The \emph{adapters}, introduced in \cref{sec:adapter_sub}, can work efficiently when the sub-VMs have no \emph{tails} (i.e., their lengths are integer numbers of words), allowing us to allocate or release $w$ bits at a time. Unfortunately, we do not know any general method that can concatenate the tails under the same performance.

\paragraph*{Toy method of concatenating tails.} To develop intuition for \emph{weak virtual memories} (introduced later), we start with a simple idea of concatenating the tails: treat them as bits and use adapters. Assume there are $B$ sub-VMs, namely $m^{[1]}, \ldots, m^{[B]}$, which we want to concatenate. Every $m^{[i]}$ can be regarded as a sequence of complete words plus a tail. We further treat the tail as a sequence of less than $w$ bits. The simple method is to concatenate the complete words using one adapter while concatenating the \emph{tail bits} using another adapter. The ``outputs'' of the two adapters are further organized into a super-VM.

This toy method has two issues. First, it only allows allocating one bit or $w$ bits at a time, but not between. For instance, if we want to allocate $w/2$ bits, we need to repeatedly allocate 1 bit for $w/2$ times, which is unaffordably slow. Second, the tail of any given sub-VM $m^{[i]}$ is not mapped to consecutive bits in the super-VM, so accessing the tail of a sub-VM could require accessing up to $\Theta(w)$ different locations in the super-VM.

Before we resolve these two issues, we first introduce the \emph{weak virtual memory} model, which appears naturally in the toy method.

\paragraph*{Weak virtual memory.}
A \emph{weak virtual memory} (WVM for short) consists of two tapes: the \emph{word tape} contains a sequence of complete words of $w$ bits each, whereas the \emph{bit tape} contains a sequence of memory bits. Similar to VM, each word (resp$.$ bit) in the word tape (resp$.$ bit tape) is labeled by a positive integer in $[1, \numword]$ (resp$.$ $[1, \numextbit]$), which we call the \emph{address} of the word (resp$.$ bit), where $\numword$ (resp$.$ $\numextbit$) is the number of words (resp$.$ bits) in the word tape (resp$.$ bit tape). The bits in the bit tape are called \emph{extra bits}, and we stipulate that the number of extra bits $\numextbit < 2w$.

In an \emph{access} operation in WVM, one is allowed to either access a word in the word tape, or access an extra bit in the bit tape, given their addresses as inputs. In an \emph{allocation} operation, one is allowed to either allocate a word in the word tape, or allocate an extra bit in the bit tape. The \emph{release} operation is defined in the similar way.

Note that WVM is a weaker model than VM: each WVM can be simulated by a VM by storing the bit tape using the tail and possibly the last complete word. Any access of the WVM can be transformed into one access of the VM, whereas any allocation of the WVM can be transformed into one allocation plus $O(1)$ accesses of the VM.

Based on this model, the above toy method of concatenating tails can be described in the following way: Assume we want to concatenate $B$ VMs $m^{[1]}, \ldots, m^{[B]}$. We treat each of them as a WVM $\tilde{m}^{[i]}$ by storing $m^{[i]}$'s tail into the bit tape of $\tilde{m}^{[i]}$ (the word tape of $\tilde{m}^{[i]}$ still stores all complete words in the VM). Then, we concatenate all word tapes using one adapter, and concatenate all bit tapes using another adapter. The outcome is stored into a super-VM.

The two issues of the toy method are closely related to WVMs: The first issue is because WVM only supports allocating or releasing a single bit or a single word at a time, but not between. The second issue arises as we treat VMs as WVMs: An access of the tail in a VM becomes $O(w)$ accesses in the WVM, which is way more expensive than accessing a complete word in the VM (the latter is transformed to still one access in the WVM).

\paragraph*{Feature on the VM size.} Under the specific framework of our algorithm, we actually do not need allocations and releases of arbitrary length: supporting $w$-bit and 1-bit allocations is enough for us, which bypasses the first issue. Recall that a level-$\l$ node $u$ stores $\numkeyu$ keys from a universe $\keyunisize{\l}$, so the data structure at node $u$ should use roughly $\log\binom{\keyunisize{\l}}{\numkeyu}$ bits of space. When adding a new key to the subtree of $u$, the number of bits to allocate is around $\log \binom{\keyunisize{\l}}{\numkeyu + 1} - \log \binom{\keyunisize{\l}}{\numkeyu}$. By choosing a proper word-length $w$, this number of bits to allocate is always $w + O(1)$, which allows us to allocate efficiently. This feature is formalized by the following claim:

\begin{claim}
  \label{clm:feature_size}
  For each $\l$ and $\nodewidth{\l}/2 \le \numkeyu \le 2\nodewidth{\l}$, we have
  \begin{gather*}
    \log \frac{\keyunisize{8h}}{\nodewidth{8h}} - 3
    \;\le\; \log \binom{\keyunisize{\l}}{\numkeyu + 1} - \log \binom{\keyunisize{\l}}{\numkeyu}
    \;\le\; \log \frac{\keyunisize{8h}}{\nodewidth{8h}} + 1;
    \numberthis \label{eq:feature_size_dif}
    \\
    \bk{\log \frac{\keyunisize{8h}}{\nodewidth{8h}} - 1} \numkeyu
    \;\le\; \log \numinstance{\l}{u}
    \;\le\; \bk{\log \frac{\keyunisize{8h}}{\nodewidth{8h}} + 3} \numkeyu.
    \numberthis \label{eq:feature_size_sum}
  \end{gather*}
\end{claim}

\begin{proof}
  One can compute
  \begin{align*}
    \log \binom{\keyunisize{\l}}{\numkeyu + 1} - \log \binom{\keyunisize{\l}}{\numkeyu} 
    = \log \frac{\keyunisize{\l} - \numkeyu}{\numkeyu + 1}.
    \numberthis \label{eq:binom_diff_simple}
  \end{align*}
  Note that $\numkeyu \ge \nodewidth{\l} / 2$, $\keyunisize{\l} - \numkeyu \ge \keyunisize{\l} / 2$, and $\numkeyu + 1 \le 2\nodewidth{\l} + 1 \le 4\nodewidth{\l}$, we have
  \begin{align*}
    \log \frac{\keyunisize{\l}}{8\nodewidth{\l}} \le \log \frac{\keyunisize{\l} - \numkeyu}{\numkeyu + 1} \le \log \frac{2\keyunisize{\l}}{\nodewidth{\l}}.
  \end{align*}
  Combining with \eqref{eq:binom_diff_simple} and $\keyunisize{\l}/\nodewidth{\l} = \keyunisize{8h}/\nodewidth{8h}$ gives \eqref{eq:feature_size_dif}.

  To show \eqref{eq:feature_size_sum}, we observe that
  \[
    \numkeyu \log \frac{\keyunisize{\l}}{\numkeyu} \le \log \binom{\keyunisize{\l}}{\numkeyu} \le \numkeyu \log \frac{\keyunisize{\l}}{\numkeyu} + (\log e) \numkeyu,
  \]
  and
  \[
    \numkeyu \log \frac{\keyunisize{\l}}{\numkeyu} 
    = \numkeyu \log \frac{\keyunisize{\l}}{\nodewidth{\l}} + \numkeyu \log \frac{\nodewidth{\l}}{\numkeyu},
  \]
  where $\keyunisize{\l}/\nodewidth{\l} = \keyunisize{8h}/\nodewidth{8h}$ and $\log \frac{\nodewidth{\l}}{\numkeyu} \in [-1, 1]$ because $\frac{1}{2} \nodewidth{\l} \le \numkeyu \le 2\nodewidth{\l}$.
Combining these two formulas together, we get
  \[
    \log \binom{\keyunisize{\l}}{\numkeyu} - \numkeyu \log \frac{\keyunisize{8h}}{\nodewidth{8h}}
    \;\in\; \Bk{ - \numkeyu, \,\tall (1 + \log e) \numkeyu}
    \;\subset\; \Bk{ - \numkeyu, \,\tall 3 \numkeyu},
  \]
  which proves \eqref{eq:feature_size_sum}.
\end{proof}

We choose $w \defeq \bigfloor{\log \frac{\keyunisize{8h}}{\nodewidth{8h}}} - 5$ as the word-length\footnote{$w$ mainly determines how we partition memory bits into words in a VM. If the ``physical'' word-length is different, one can simulate the algorithm with word-length $w$ at no additional cost since $w \le \log \keyunisize{8h} = O(\log N)$.}, then:
\begin{enumerate}[label=\textup{(\arabic*)}]
\item\label{enum:feature_dif} inserting (deleting) a key to $u$ will cause an allocation (release) of one word plus $O(1)$ bits;
\item\label{enum:feature_sum} the space usage of $u$'s data structure roughly equals $\numkeyu$ words plus $O(\numkeyu)$ bits.
\end{enumerate}
Here, Condition~\ref{enum:feature_dif} addresses the first issue of the toy method; Condition~\ref{enum:feature_sum} will lead to the following property and benefit resolving the second issue.

\begin{definition}[Word-dominant]
  \label{def:word_dominant}
  A VM is \emph{word-dominant} if at any given time, there is a positive integer $n$, such that the current number of bits in the VM lies in $[nw + 3n, \, nw + 10n]$. For a word-dominant VM, we forbid allocations and releases that will result in a violation of this word-dominant condition.
\end{definition}

\paragraph*{Random swap.} As mentioned in the second issue, when we treat a VM as a WVM by storing the tail in the bit tape, accessing the tail of VM becomes very slow, as it would require $O(w)$ bit-accesses to the WVM. However, accessing any given complete word of VM is still efficient (transformed into only one word-access to the WVM). To resolve the second issue, we amortize the cost of tail-access by swapping the tail with a random complete word. Specifically, we choose a word $i \in [\numword]$ almost uniformly at random, and swap the tail with the prefix of word $i$ of the same length $\numextbit$ before we store the ``new tail'' into the bit tape of WVM. As result, accessing the original tail of VM becomes efficient, since we only need to access one word in the WVM; accessing the original word $i$ requires us to access the new tail, which costs $O(w)$ accesses to the WVM. Since $i$ is randomly chosen, the expected cost of any given access operation is low. This idea is formalized as the following lemma.

\begin{restatable}[Random swap]{lemma}{randomswap}
  \label{lm:random_swap}
  We can store a word-dominant VM containing $M$ bits ($2w < M < Nw$) into a WVM with the same total number of bits $w \numword + \numextbit = M$, while any given operation to the VM can be transformed into several operations to the WVM within constant expected additional computation time. Moreover:
  \begin{enumerate}[label=\textup{(\alph*)}]
  \item\label{enum:access_of_WVM} An access to the VM (either a complete word or the tail) is transformed into $1 + O(\min(1, w^2/M))$ accesses to the WVM in expectation.
  \item\label{enum:allocation_of_WVM} Allocating/releasing a word or a bit in the VM can be transformed into $O(1)$ allocation/release operations followed by $O(1)$ accesses to the WVM (both in expectation).
  \item\label{enum:no_extra_space} The transformation from VM to WVM does not introduce additional space usage, except that we need to access an $O(\log^2 N)$-bit random seed which can be shared between multiple instances. $\numword$ and $\numextbit$ are fully determined by $M$ and the seed.
  \end{enumerate}
\end{restatable}

This lemma shows that word-dominant VMs can be ``simulated'' by WVMs without much loss of efficiency. Its proof is deferred to \cref{sec:randomswap}.

\subsection{Inductive Construction of Adapter Tree}
\label{sec:induction_ada}

In this subsection, we will construct the core part of our dictionary -- an adapter tree storing the set of keys in a level-1 bucket -- by inductively aggregating the sub-VMs of children. Formally, we will prove the following statement inductively:

\begin{statement}
  \label{statement:induction}
  Fix a level-$\l$ node $u$ and denote the number of keys hashed to $u$ by $\numkeyu$. Fix parameter $r = \nodewidth{1}$.
  Assuming free access to $\numkeyu$, there is a data structure maintain the set of keys hashed to $u$ with spill universe $[\spilluni{\l}{u}]$ and $\nummembits{\l}{u}$ memory bits, such that 
  \begin{itemize}
  \item $r < \spilluni{\l}{u} \le 2 r$.
  \item $\log \numinstance{\l}{u} - 1 < \nummembits{\l}{u} + \log \spilluni{\l}{u} \,\le\, \log \numinstance{\l}{u} + \frac{\nodewidth{\l} - 1}{r}$.
  \end{itemize}
  Moreover, the $\nummembits{\l}{u}$ memory bits are stored in a VM.
\end{statement}

We assume without loss of generality that the lower bounds $\nummembits{\l}{u} + \log \spilluni{\l}{u} > \log \numinstance{\l}{u} - 1$ and $\spilluni{\l}{u} > r$ hold, because otherwise we can pad zeros to the end of the memory until $\nummembits{\l}{u} + \log \spilluni{\l}{u} > \log \numinstance{\l}{u} - 1$, or repeatedly include the last memory bit into the spillover\footnote{We would never run out of memory bits as $\log 2r \ll \log \numinstance{\l}{u} \approx \nummembits{\l}{u} + \log \spilluni{\l}{u}$.} until $r < \spilluni{\l}{u} \le 2r$. Hence, we will not prove these two inequalities in the following inductive proof.

\begin{proofof}{\cref{statement:induction}}
  We prove this by induction. The proof is similar to \cite{li2023dynamic} and \cite{patrascu2008succincter}, except that the WVMs and the random swap lemma will involve when we aggregate sub-VMs.

  \paragraph*{Base case.} When $\l = 8h$, for each level-$8h$ node $u$, there are $\numkeyu$ keys hashed to it, each is represented as an ID-quotient pair $(\yid, \xquot)$. By regarding the ID as ``key'' and quotient as ``value'', we use a minimap to maintain it:
  The set of IDs can be encoded into
  \[
    \numkeyu \cdot \idleni[8h]
    \,\le\,
    \nodewidthub{8h} \cdot \idleni[1]
    \,=\,
    \Theta(B \cdot \log \nodewidth{1})
    \,=\,
    \Theta(\sqrt{\log N} \cdot h \cdot \log \log N)
    \,\ll\, O(w)
  \]
  bits. (Recall that $\idleni[1] = \ceil{4 \log \nodewidth{1}}$ and $\nodewidth{1} \le 2 B^{8h}$; the last inequality holds because $h \le \log^{0.3}N$.)
  Applying \cref{lm:minimap_lem} with key-universe $[2^{\idleni[8h]}]$, value-universe $[\quotientuni]$ and the same $r$ as here, we get a minimap which maintains the set of keys hashed to $u$, with query time $O(1)$, insertion/deletion time $O(\log \numkeyu) = O(\log \log N)$ (here $\numkeyu \le \nodewidthub{8h} \le 2B$), spill universe $K \le 2r$, and redundancy at most $\frac{4}{r} \le \frac{\nodewidth{\l} - 1}{r}$ bits.

  \paragraph*{Induction step.} Assume the induction hypothesis holds for level $\l + 1$ and we are going to prove it for level $\l$. For any level-$\l$ node $u$, denote its $B$ children by $v_1, v_2, \ldots, v_B$; denote by $\numkeyvi$ the number of keys hashed to $v_i$. By the induction hypothesis, these $\numkeyvi$ keys can be represented by $\nummembits{\l+1}{v_i}$ memory bits and a spillover $\spillvi \in [\spilluni{\l+1}{v_i}]$, with these $\nummembits{\l+1}{v_i}$ memory bits stored in a sub-VM from $v_i$. Below, we will introduce our encoding procedure step by step.

  \paragraph*{Step 1: Random swap.} Guided by the intuition introduced in \cref{sec:wvm}, the first step is to transform the sub-VM from each $v_i$ into a WVM. To apply \cref{lm:random_swap}, we only need to check that these sub-VMs are word-dominant. Recalling that $w + 5 \le \log \frac{\keyunisize{8h}}{\nodewidth{8h}} < w + 6$, by \cref{clm:feature_size}, we have
  \begin{align*}
    \numkeyvi w + 4\numkeyvi\le \log \numinstance{\l+1}{v_i} \le \numkeyvi w + 9\numkeyvi.
  \end{align*}
  Hence by the induction hypothesis, we have
  \begin{align*}
    \nummembits{\l+1}{v_i}
    \le{}& \nummembits{\l+1}{v_i} + \log \spilluni{\l+1}{v_i}\\
    \le{}& \log \numinstance{\l+1}{v_i} + \frac{\nodewidth{\l+1} - 1}{r}
          \le{} \numkeyvi w + 10 \numkeyvi
  \end{align*}
  (the last inequality holds because $\nodewidth{\l+1} - 1 < \nodewidth{1} = r < \numkeyvi r$).
  On the other hand, by induction hypothesis, we have $\nummembits{\l + 1}{v_i} + \log \spilluni{\l + 1}{v_i} \ge \log \numinstance{\l + 1}{v_i} - 1$.
  Hence,
  \begin{align*}
    \nummembits{\l+1}{v_i}
    \ge{}& \log \numinstance{\l + 1}{v_i} - \log \spilluni{\l + 1}{v_i} - 1 \\
    \ge{}& \numkeyvi w + 4\numkeyvi - \log (2r) - 1
          \ge{} \numkeyvi w + 3\numkeyvi
  \end{align*}
  (the last inequality holds because $\log 2r + 1 = \Theta(h \log \log N) \ll B/2 \le \nodewidthlb{8h} \le \numkeyvi$, as $h \le \log^{0.3} N$ and $B \in [\sqrt{\log N}, \, 2\sqrt{\log N})$).
  Therefore, the sub-VM for each $v_i$ is word-dominant (see \cref{def:word_dominant}), and by \cref{lm:random_swap}, it can be stored into a WVM with little performance loss.

  \paragraph*{Step 2: Adapters.} The WVM from each $v_i$, according to its definition, is composed of a tape of $\numwordi$ words (denoted by $\completewordi$), and a tape of $\numextbiti$ extra bits (denoted by $\extrabitsi$). We use two adapters to aggregate all of them together.

  The first adapter aggregates all the extra bits $\bigBK{\extrabitsi : 1 \le i \le B}$. Although in \cref{sec:adapter_sub} we only considered how to aggregate sequences of complete words, by treating every bit as ``a word with word-length 1'', we can easily construct a \emph{bit-wise} adapter that aggregates sequences of bits. The output of this adapter is a sequence of $\sum_{i=1}^B \numextbit^{[i]}$ bits, which we store into a VM called the \emph{bit VM}. We denote the bit VM's complete-word part and tail part as $\completewordi[0]$ and $\extrabitsi[0]$, respectively.

  The second adapter aggregates all the complete words $\bigBK{\completewordi: 0 \le i \le B}$, including the complete words produced by the first adapter. This results in a VM with $\sum_{i=1}^B \numwordi + \bigfloor{\bigbk{\sum_{i=1}^B \numextbit} \big/ w}$ words, which we call the \emph{word VM}.

  Finally, we directly concatenate the word VM with the tail part of the bit VM, getting a VM with totally $\Mcat \defeq \sum_{i=1}^B \bigbk{w\numwordi + \numextbiti} = \sum_{i=1}^B \nummembits{\l+1}{v_i}$ bits, which we call the \emph{concatenated memory}, namely, $\mcat$.

  \paragraph*{Step 3: Cut the memory.} We cut the concatenated memory $\mcat$ into two parts $\mfix$ and $\mrem$, such that the first part has $\Mfix$ bits which only depends on $\numkeyu$ but \emph{not} $\numkeyvi$ for each $i$; the second part has at most $O(w)$ bits. Formally, we define
  \[
    \Mmax \defeq \max_{\numkeyvi[1] + \numkeyvi[2] + \cdots + \numkeyvi[B] = \numkeyu} \sum_{i=1}^B \nummembits{\l+1}{v_i},
    \qquad
    \Mfix \defeq \Mmax - w.
  \]
  (There is always $\Mfix \ge 0$ because $\Mmax > \nummembits{\l + 1}{v_i} = \Theta\bigbk{\log \binom{\keyunisize{\l + 1}}{\nodewidth{\l + 1}}} - \Theta(\log r) \ge \Theta(Bw) \gg w$.)
  We divide $\mcat$ into the leftmost $\Mfix$ bits and the remaining $\Mrem \defeq \Mcat - \Mfix$ bits. It is possible that $\Mcat < \Mfix$ for the current number of keys $\numkeyvi$, in which case $\mfix$ is formed by padding zeros to the end of $\mcat$ until it has $\Mfix$ bits; $\mrem$ is left empty. In all cases, the second part contains at most $w$ bits.

  After cutting the memory into two parts, $\mfix$ directly appears as the leftmost bits in our final encoding, while $\mrem$ is further compressed with other information in the next step.

  \paragraph{Step 4: Compress the labels and spillovers.} Next, we compress the children's numbers of keys $\numkeyvec \defeq \bk{\numkeyvi[1], \numkeyvi[2], \ldots, \numkeyvi[B]}$, their spillovers $\bk{\spillvi[1], \spillvi[2], \ldots, \spillvi[B]}$, and the remaining part $\mrem$ from the last step together, using the following lemma from~\cite{patrascu2008succincter}.

  \begin{lemma}[\cite{patrascu2008succincter}]
    \label{lm:encode_spillover}
    Assume we have to represent a variable $x \in \mathcal{X}$, and a pair $(\ym, \yk) \in \BK{0,1}^{M(x)} \times [K(x)]$. Let $p(x)$ be a probability density function on $\mathcal{X}$, and $K(\cdot)$, $M(\cdot)$ be non-negative functions on $\mathcal{X}$ satisfying:
    \begin{align*}
      \label{eq:pat_lm5_condition}
      \forall x \in \mathcal{X} :\quad \log \frac{1}{p(x)} + M(x) + \log K(x) \le H. \numberthis
    \end{align*}
    We further assume the word size $w = \Omega(\log |\mathcal{X}| + \log r + \log \max K(x))$, then we can design a spillover representation for $x$, $\ym$, and $\yk$, denoted by $(m^*, k^*)\in \{0,1\}^{M^*}\times [K^*]$, with the following parameters:
    \begin{itemize}
    \item The spill universe is $K^*$ with $K^* \le 2r$; the memory size is $M^*$ bits.
    \item The redundancy is at most $4/r$ bits, i.e., $M^* + \log K^* \le H + 4/r$.
    \item Given a precomputed table of $O(|\X| \cdot r \cdot \max K(x))$ words that only depends on the input functions $K, M$, and $p$, and assuming $H \le O(w)$, both decoding $(x, \ym, \yk)$ from $(m^*, k^*)$ and encoding $(x, \ym, \yk)$ to $(m^*, k^*)$ takes $O(1)$ time on a word RAM. The table can be precomputed in linear time.
    \end{itemize}
  \end{lemma}

  Let $\kcat \in \bigBk{\prod_{i=1}^B \spilluni{\l+1}{v_i}}$ be the combination of the children's spillovers $\bk{\spillvi[1], \spillvi[2], \ldots, \spillvi[B]}$. We are going to apply the above lemma on $\X = \midBK{\numkeyvec : \sum_{i=1}^B \numkeyvi = \numkeyu}$ and $(\ym, \yk) = (\mrem, \kcat)$.

  To construct the probability distribution, we first define
  \[
    p(\numkeyvec) \defeq \frac{\prod_{i=1}^B \numinstance{\l+1}{v_i}}{\numinstance{\l}{u}},
  \]
  that is, the induced marginal distribution on $\numkeyvec$ when the set of keys hashed to $u$ follows the uniform distribution over all $\numkeyu$-element subsets of the universe. Same as \cite{patrascu2008succincter}, we slightly perturb the distribution for stronger properties.
  
  \begin{claim}[\cite{patrascu2008succincter}]
    \label{clm:perturb}
    For any probability distribution $p(\cdot)$ over set $\X$ and any parameter $r > 0$, we can perturb $p(\cdot)$ to another probability distribution $p'(\cdot)$, such that for any $x \in \X$,
    \begin{itemize}
    \item $p'(x) \ge \frac{1}{2r|\X|}$.
    \item $\log \frac{1}{p'(x)} \le \log \frac{1}{p(x)} + \frac{2}{r}$.
    \end{itemize}
  \end{claim}
  
  We then apply \cref{lm:encode_spillover} with the perturbed distribution $p'$ and
  \[H \defeq \log \numinstance{\l}{u} - \Mfix + \frac{\nodewidth{\l} - 5}{r}.\]
  We check the condition \eqref{eq:pat_lm5_condition} by discussing two cases:
  \begin{itemize}
  \item Suppose $\Mcat \ge \Mfix$, i.e., we did not pad zeros to $\Mcat$. In this case, $|\mrem| = \Mcat - \Mfix$. The induction hypothesis implies that
    \[
      \nummembits{\l+1}{v_i} + \log \spilluni{\l+1}{v_i} \le \log \numinstance{\l+1}{v_i} + \frac{\nodewidth{\l+1} - 1}{r}.
    \]
    Therefore, for any $\numkeyvec \in \X$, the LHS of \eqref{eq:pat_lm5_condition} is
    \begin{align*}
      & \phantom{{}\le{}} \log \frac{1}{p'(\numkeyvec)} + (\Mcat(\numkeyvec) - \Mfix) + \log \bk{\prod_{i=1}^B \spilluni{\l+1}{v_i}} \\
      & \le
        \bk{\log \frac{1}{p(\numkeyvec)} + \frac{2}{r}} + \sum_{i=1}^B \nummembits{\l+1}{v_i} - \Mfix + \sum_{i=1}^B \log \spilluni{\l+1}{v_i} \\
      & \le
        \log {\frac{\numinstance{\l}{u}}{\prod_{i=1}^B \numinstance{\l+1}{v_i}}} + 
        \sum_{i=1}^B \bk{\nummembits{\l+1}{v_i} + \log \spilluni{\l+1}{v_i}}  + \frac{2}{r} - \Mfix \\
      & \le
        \log \numinstance{\l}{u} + B \cdot \frac{\nodewidth{\l+1} - 1}{r} + \frac{2}{r} - \Mfix \\
      & =
        \log \numinstance{\l}{u} + \frac{\nodewidth{\l} - B + 2}{r} - \Mfix \\
      & < H.
    \end{align*}
  \item Suppose $\Mcat < \Mfix$. In this case, $\mrem$ is empty, so the LHS of \eqref{eq:pat_lm5_condition} equals
    \begin{align*}
      & \phantom{{} \le {}}
        \log \frac{1}{p'(\numkeyvec)} + \log \bk{\prod_{i=1}^B \spilluni{\l+1}{v_i}} \\
      & \le \log (2r |\X|) + \log (2r)^B = \Theta(B \log r) \ll w,
    \end{align*}
    where the first inequality is due to \cref{clm:perturb} and the induction hypothesis $\spilluni{\l+1}{v_i} \le 2r$; the second inequality holds as $|\X| \le \numkeyu^B \le (2\nodewidth{1})^B = (2r)^B$; the last inequality holds as $B \log r = \Theta(\sqrt{\log N} \cdot h \log \log N) \ll w$, since $B\le 2\sqrt{\log N}$ and $h \le \log^{0.3} N$. On the other side, the RHS of \eqref{eq:pat_lm5_condition} is
    \begin{align*}
      H &\ge \max_{\numkeyvec \in \X} \log \bk{\prod_{i=1}^B \numinstance{\l+1}{v_i}} - \Mfix + \frac{\nodewidth{\l} - 5}{r} \\
        &\ge \max_{\numkeyvec \in \X} \bk{\Mcat(\numkeyvec) - B \cdot \frac{\nodewidth{\l+1} - 1}{r}} - \Mfix + \frac{\nodewidth{\l} - 5}{r} \\
        &= \Mmax - \Mfix + \frac{B - 5}{r} \;>\; w \;\ge\; \textup{LHS}.
    \end{align*}
    Here, the second inequality holds due to the redundancy constraint in the induction hypothesis; the third inequality is because $\Mfix = \Mmax - w$. Therefore \eqref{eq:pat_lm5_condition} also holds in this case.
  \end{itemize}
  Applying \cref{lm:encode_spillover} gives us a spillover representation $(m^*, k^*) \in \midBK{0, 1}^{M^*} \times [K^*]$ of $(\numkeyvec, \kcat, \mrem)$ conditioning on $\numkeyu$, where $K^* \le 2r$ and
  \[
    M^* + \log K^* \le H + \frac{4}{r} = \log \numinstance{\l}{u} + \frac{\nodewidth{\l} - 1}{r} - \Mfix.
  \]
  With the help of proper lookup tables, both the encoding and decoding procedures can be completed within constant time.
  
  \paragraph{Step 5: Concatenate.} The last step involves concatenating $\mfix$ with $m^*$, the outcome memory bits obtained from the previous compression step, to form a bit string of length $\Mfix + M^* \eqdef \nummembits{\l}{u}$. This memory string, combined with the spillover $k^* \in [K^*] = [\spilluni{\l}{u}]$, form the encoding for the set of keys hashed to level-$\l$ node $u$. The induction hypothesis holds for $\l$, since $K^*\leq 2r$ and
  \[
    \nummembits{\l}{u} + \log \spilluni{\l}{u} = \Mfix + M^* + \log K^* \le \log \numinstance{\l}{u} + \frac{\nodewidth{\l} - 1}{r}.
  \]
  Finally, the $\nummembits{\l}{u}$ memory bits are again divided into $\floor{\nummembits{\l}{u} / w}$ words and a tail, stored within a VM.
\end{proofof}
  
\begin{remark}
  Observe that the primary part, $\mfix$, which comes from the adapter, is ``aligned'' with the resulting VM. This means that a complete word in a child's sub-VM is still stored as a complete word in the new VM. This alignment benefits our query and update algorithms because each word-access from a child will translate into only one word-access of $u$'s virtual memory.
\end{remark}

Now, consider Statement~\ref{statement:induction} for level-1 node $u$. We directly encode the final spillover $k \in [\spilluni{1}{u}]$ into memory, incurring a 1-bit redundancy due to rounding. Then, we can store the set of keys hashed to a level-1 bucket within a space of
\begin{align*}
  \nummembits{1}{u} + \log \spilluni{1}{u} + 1 \le \log \numinstance{1}{u} + \frac{\nodewidth{1} - 1}{r} + 1 < \log \numinstance{1}{u} + 2.
\end{align*}

\subsection{Construction of the whole Dictionary}

The whole dictionary consists of the following parts:
\begin{itemize}
\item A fixed space of $\log \binom{\keyunisize{1}}{\nodewidthub{1}}+2$ bits for each level-1 node. As we have seen, the set of keys hashed to a level-1 bucket $u$ can be stored in $\log \numinstance{1}{u} + 2 \le \log \binom{\keyunisize{1}}{\nodewidthub{1}}+2$ bits, assuming free access to the number of keys $\numkeyu$.
\item A space of $\log \nodewidthub{1}$ bits for each level-1 node to store the number of keys hashed to it.
\item Several lookup tables used in the induction processes, which will be analyzed later.
\item Random seed used in the random swap lemma (\cref{lm:random_swap}), which is shared between nodes on the same level, while the seeds for different levels are chosen independently. They only occupy a very limited space of $O(h \cdot \log^2 N)$ bits.
\item The space for storing hash functions, which is $o(R)$ bits. The analysis is deferred to \cref{sec:hashfunc_for_adapter}.
\end{itemize}
As mentioned in \cref{sec:intro}, every part of the memory only depends on the current set of keys and some random bits, which means that our dictionary is strongly history-independent. Below, we calculate the space usage of each part.

\smallskip

The first part occupies the most of the space. Recalling that $R \defeq N / \log^h N = \Omega\bigbk{N \big/ \nodewidth{1}^{1/4}} \gg (N \log N) \big/ \nodewidth{1}^{1/3}$, we have
\begin{align*}
  \frac{N}{\nodewidth{1}} \log \binom{\keyunisize{1}}{\nodewidthub{1}}
  &\le \log \binom{U}{N \cdot \nodewidthub{1} / \nodewidth{1}}
    \le \log \binom{U}{N \cdot \bigbk{1 + 2 \cdot \nodewidth{1}^{-1/3}}}
    \le \bigbk{1 + 2 \cdot \nodewidth{1}^{-1/3}} \log \binom{U}{N} \\
  &= \log \binom{U}{N} + \Theta\bigbk{(N \log U) \big/ \nodewidth{1}^{1/3}}
    \le \log \binom{U}{N} + O(R),
\end{align*}
where the third inequality comes from the log concavity of $\binom{n}{k}$ as a function of $k$. Therefore, the total space usage of the first part does not exceed $\frac{N}{\nodewidth{1}} \log \binom{\keyunisize{1}}{\nodewidthub{1}} + \frac{2N}{\nodewidth{1}} \le \log \binom{U}{N} + O(R)$ bits.

The space usage of the second and the fourth parts do not exceed $\frac{N}{\nodewidth{1}} \cdot \log N = o(R)$ bits and $O(\log^3 N) \ll o(R)$ bits, respectively. The hash functions (the last part) only occupy $o(R)$ bits as well.

It remains to check the sizes of the lookup tables, which are listed below. All the lookup tables occupy at most $O\bigbk{\sqrt{N}}$ words which is far less than $O(R)$ bits.
\begin{itemize}
\item lookup table for minimaps. We used minimaps to maintain at most $n' \defeq \nodewidthub{8h} = \Theta\bigbk{\sqrt{\log N}}$ key-value pairs from a key-universe $U' \defeq 2^{\idleni[8h]} = \poly \log N$.
  According to \cref{lm:minimap_lem}, the lookup table size is
  $U' n' \binom{U'}{n'} \cdot (2rn')^{n'} = (\poly \log N)^{\Theta(\sqrt{\log N})} \ll O(\sqrt{N})$
  words.

\item lookup table for adapters. According to \cref{lm:adapter}, the lookup table for adapters occupies $O\bigbk{B \Lmax^{B+1}}$ words, where $\Lmax$ is the maximum number of words in the aggregated super-VM. Plugging in $\Lmax = O(\nodewidth{1}) = \bk{\log N}^{O(\log^{0.3}N)}$ and $B = \Theta(\sqrt{\log N})$, we know the lookup table size is less than $O(\sqrt{N})$ words.

\item lookup table for encoding/decoding $\numkeyvec$ and $\kcat$. At the beginning of the step 4, we need to encode the number of keys in each child $\numkeyvec = \bk{\numkeyvi[1], \numkeyvi[2], \ldots, \numkeyvi[B]}$ using its index in the set $\bigBK{\numkeyvec: \sum_{i=1}^B \numkeyvi = \numkeyu}$. For a fixed $\l$ and $\numkeyu$, it needs a lookup table of $O\bigbk{\numkeyu^B}$ words. Hence, the total size is at most $O(h \nodewidthub{1}^{B+1})$ words.

  Moreover, we need to compute vectors $\bigbk{\numwordi}_{i \in [B]}$ and $\bigbk{\numextbiti}_{i \in [B]}$, the sizes of the tapes of WVMs from all children, which are required for adapter operations and are fully determined by $\numkeyvec$. This is done with a lookup table of the same size.

  We also need to encode all the spillovers $\spillvi[1], \spillvi[2], \ldots, \spillvi[B]$ into $\kcat \in \bigBk{\prod_{i=1}^B \spilluni{\l+1}{v_i}}$. For a fixed $\l$ and $\numkeyvec$, it needs a lookup table of $O\bigbk{(2r)^B}$ words, so the total size is at most $O\bigbk{h \nodewidthub{1}^{B} (2r)^B} = (\log N)^{O(\log^{0.3} N) \cdot \Theta(\sqrt{\log N})} \ll O(\sqrt{N})$ words.

\item lookup table for \cref{lm:encode_spillover}. For each level $\l$ and each possible number of keys $\numkeyu$ in a node $u$, we need a table of $O(|\X| \cdot r \cdot \max K(x))$ words, where $\X$ is the set formed by the array of number of keys in children of $u$, i.e., $\X \defeq \midBK{(\numkeyvi[1], \numkeyvi[2], \ldots, \numkeyvi[B]) : \sum_{i=1}^B \numkeyvi = \numkeyu}$, whose size is at most $O\bigbk{\nodewidthub{1}^B}$; we also have $\max K(x) \le (2r)^{B}$. Taking summation over all possible $n(u)$'s and all $h$ levels, the total size is $O\bigbk{h \cdot \nodewidthub{1} \cdot \nodewidthub{1}^{B} \cdot r \cdot (2r)^{B}} = (\log N)^{O(\log^{0.3} N) \cdot \Theta(\sqrt{\log N})} \ll O(\sqrt{N}) \ll O(\sqrt{N})$ words.

\item Tables for quickly computing each $\nummembits{\l}{u}$, $\spilluni{\l}{v_i}$, and $\Mfix = \Mfix\bk{\keyunisize{\l}, \numkeyu}$. They occupy $O(h\nodewidthub{1}) \ll O(\sqrt N)$ words.
\end{itemize}
All lookup tables mentioned above can be pre-computed efficiently and will not become the time bottleneck. Now, summing up the space usage of all parts, the data structure occupies a space of $\log \binom{U}{N} + O(R)$ bits.

\subsection{Query, Insertion, and Deletion}

All operations on the data structure begin by simulating the hashing process on the input key $x$. Specifically, we first apply the first hash function $h$ on $x$, obtaining a triple $(u, \, \xid, \, \xquot)$, which means $x$ is hashed to a root node $u$ with the ID-quotient pair $(\xid, \, \xquot)$. Then, by applying the second hash function, $x$ is further hashed to a (level-$8h$) leaf node $v$, with ID-quotient pair $(\yid, \, \xquot)$ in $v$. These nodes indicate a path from the root to the leaf, which we will walk along in all types of operations.

\paragraph*{Query algorithm and nested adapters.}

Recall that every node $u$ maintains a sub-data structure storing keys hashed into $u$'s subtree, which is stored in a VM and a spillover $k = k^*$. The VM consists of two parts: $\mfix$ and $m^*$, where $\mfix$ comes from adapters that concatenate sub-VMs from the children, and $m^*$ includes the information (spillovers and numbers of keys) of the children. The latter part $m^*$ has at most $O(w)$ bits and is contained in the rightmost $O(1)$ words in $u$'s VM. Throughout this section, we call $m^*$ the \emph{metadata region} of $u$'s VM.

Before we visit the node $u$, we need the knowledge of its spillover $k$ and number of keys $\numkeyu$. Then, we immediately read the metadata $m^*$ on $u$'s VM, which takes $O(1)$ word-accesses.
According to \cref{lm:encode_spillover}, we can decode $(k^*, m^*)$ in constant time to recover the number of keys in each child, $\numkeyvi[1], \numkeyvi[2], \ldots, \numkeyvi[B]$, their spillovers $k_1, k_2, \ldots, k_B$, and the rightmost bits of the concatenated memory, $\mrem$. Consequently, we get vectors $\bigbk{\numwordi}_{i \in [B]}$ and $\bigbk{\numextbiti}_{i \in [B]}$, the sizes of both tapes of WVMs from children, which are fully determined by $\numkeyvec$. Next, the query algorithm recurses into a child of $u$ along the path, and repeat the above process until it reaches the leaf. Finally, it makes a query to the minimap on the leaf to get the answer.

The above description has assumed direct access to the VM of any given node. However, only the VM of the root node is directly stored into the ``physical memory'', while the VM of other nodes are connected to it via several nested adapters (and random-swap structures (\cref{lm:random_swap}) which stores VMs into WVMs). To access a word in node $u$'s VM, we first translate it into access requests to the parent's VM, then to the grandparent, and continue until the root is reached.
Formally, we use the following subroutine (\cref{alg:probe_word}) to access a word in the VM of an arbitrary node. We denote by $(v, i)$ the $i$-th word in the VM of node $v$. As reading the tail of the VM has the same interface as reading a word, we do not distinguish them and regard the tail as the last word.

\begin{algorithm}[ht]
    \caption{Accessing Virtual Memory Words}
    \label{alg:probe_word}
    \DontPrintSemicolon

    \SetKwFunction{fprobe}{Access}
    \SetKwProg{Fn}{Function}{:}{}

    \Fn(\Comment{Access the $i$-th word in $v$'s VM}){\fprobe{$v$, $i$}} {
      \If{$v$ is the root} {
        Directly access the $i$-th word\;
        \Return
      }
      $u \gets v$'s parent\;
      Translate the access request $(v, i)$ into access requests to $v$'s WVM according to the random swap lemma (\cref{lm:random_swap})\;
      Translate each WVM access request into an access request to $u$'s \emph{concatenated memory} $\mcat = \mfix \concat \mrem$ via the adapter\;
      \ForEach{access request to the $j$-th word in $u$'s concatenated memory} {
        \If{any part of word $j$ belongs to $\mfix$} {
          \fprobe{$u, j$}\;
        }
        \If{any part of word $j$ belongs to $\mrem$} {
          Access the corresponding bits in $\mrem$ which has been decoded via \cref{lm:encode_spillover}\;
          \label{line:decode}
        }
      }
    }
\end{algorithm}

When we read a word $(v, i)$, \cref{line:decode} does not produce extra word-accesses, because we have already gained the knowledge of $\mrem$ when we visit $u$ (we always need to visit the parent $u$ before we can visit $v$). On the other hand, when we write to some word $(v, i)$ (which is not required for queries, but we will need it for updates), it might seem that \cref{line:decode} would need to write to $\mrem$ of $u$; what we actually do here is to delay the writing request to $\mrem$ until the end of the entire update operation, at which point we will update all $\mrem$ bottom-up along the path we visited.

Due to the random swap lemma (\cref{lm:random_swap}), accessing $(v,i)$ may cause more than 1 word-accesses to the VM on the parent of $u$. Specifically, assuming $v$ is in level $\l$, the expected number of word-accesses to its parent $u$ is at most
\[
  1 + O\bk{\min\bk{1, w^2/\nummembits{\l}{v}}} = 1 + O(\min(1, w / \numkey{v})) \le 1 + c\min\bk{1, B^{1-8h+\l}}
\]
for a fixed constant $c > 0$ (note that $\numkey{v} = \Theta(\nodewidth{\l}) = \Omega(B^{8h - \l + 1})$ and $w = \Theta(B^2)$).
As the random seeds for random swapping are independent for different levels, and
\[
  \prod_{\l = 1}^{8h} \bk{1 + c\min\bk{1, B^{1-8h+\l}}} \le (1 + c)^2 \cdot \prod_{i = 1}^{\infty} \bk{1 + \frac{c}{B^i}} = O(1),
\]
we can see that every word-access in node $v$ will cause at most $O(1)$ word-accesses to the root in expectation, and it takes $O(\l) = O(h)$ time to translate the word-access requests from $v$ to the root. That means we need $O(h)$ expected time to access a word in the VM of any given node. Moreover, the query algorithm initiates $O(1)$ word-accesses at every node it visits (mainly for reading $m^*$), so the total running time for the query algorithm is $O(h^2)$.

\paragraph*{Insertion and deletion algorithms.}

For insertions and deletions, we first follow the same procedure as the query algorithm, walking from the root down to the leaf that the key was hashed to, recovering the numbers of keys $\numkeyu$ and spillovers for all the nodes along the path. Next, we insert/delete the key in the minimap on the leaf, which takes $O(\log \nodewidthub{8h}) = O(\log \log N)$ word-accesses to the VM of the leaf.

However, when we insert/delete a key in the leaf, every node $u$ on the path changes its number of keys $\numkeyu$, and therefore changes the space usage of $u$'s VM, i.e., $\nummembits{\l}{u}$. Additional time is spent to adjust the VM sizes. Due to the symmetry, we only analyze the insertions below. We will show later that the difference of $\nummembits{\l}{u}$ before and after the insertion is $w \pm O(1)$, which can be expressed as $O(1)$ allocation/release requests to $u$'s VM, each allocating/releasing exactly $w$ bits or a single bit. According to \cref{lm:random_swap}, we can further translate these requests to $O(1)$ allocation/release requests (followed by $O(1)$ word or bit accesses) to the WVM, which are further translated to $O(1)$ allocation/release requests to the two adapters connecting $u$'s WVM with its parent. They take $O(\log \max(\nodewidth{\l}, 2w)) = O((8h - \l) \cdot \log \log N)$ word-accesses to the VM on the parent of $u$, each can be completed within $O(h)$ time as shown above. Taking summation over all $h$ nodes on the path, the time complexity for an insertion is $O(h^3 \log \log N)$.

After adjusting the VM sizes, for all nodes $u$ on the path from bottom to top, we redo all encoding steps introduced in \cref{sec:induction_ada}, and update all changed words. They may include:
\begin{itemize}
\item Step 4 involves compressing $O(1)$ words of information into a spillover representation. Its resulting memory bits are directly stored in the rightmost words of $u$'s VM, namely $m^*$. We initiate $O(1)$ word-accesses to rewrite all of them.
\item Changing the number of keys $\numkeyu$ of node $u$ may cause $\Mfix$ to change since it depends on $\numkeyu$. $\Mfix$ will change by at most $O(w)$ (bits) since it only differs from $\nummembits{\l}{u}$ by $O(w)$, while the latter only changes by $O(w)$ during an update. The change of $\Mfix$ requires us to update the rightmost $O(1)$ words in $\mfix$, incurring no more than $O(1)$ word-accesses at each level.
\end{itemize}
The time spent on these word-updates is less than the time spent for allocations, which is $O(h^3 \log \log N)$ as shown above. Therefore, the (expected) time complexity for an insertion or a deletion is $O(h^3 \log \log N)$.

\bigskip

It remains to show that the difference of VM size before and after the insertion is equal to $w \pm O(1)$ bits, i.e.,
\begin{align*}
  w - O(1) \le M(\keyunisize{\l}, \numkeyu + 1) - \nummembits{\l}{u} \le w + O(1).
  \numberthis \label{eq:few_alloc}
\end{align*}
Actually, by the induction hypothesis, we have
\begin{gather*}
  \log \numinstance{\l}{u} - 1 < \nummembits{\l}{u} + \log \spilluni{\l}{u} \le \log \numinstance{\l}{u} + 1, \\
  \log \binom{\keyunisize{\l}}{\numkeyu + 1} - 1 < M(\keyunisize{\l}, \numkeyu + 1) + \log K(\keyunisize{\l}, \numkeyu + 1) \le \log \binom{\keyunisize{\l}}{\numkeyu + 1} + 1.
\end{gather*}
As $\spilluni{\l}{u}, K(\keyunisize{\l}, \numkeyu + 1) \in (r, 2r]$, we can see that $\abs{\log \spilluni{\l}{u} - \log K(\keyunisize{\l}, \numkeyu + 1)} < 1$. Moreover, by \cref{clm:feature_size} and $w + 5 \le \log \frac{\keyunisize{8h}}{\nodewidth{8h}} < w + 6$, we have
\[
  w + 2 \le \log \binom{\keyunisize{\l}}{\numkeyu + 1} - \log \binom{\keyunisize{\l}}{\numkeyu} \le w + 7.
\]
Combining these inequalities, we get
\begin{align*}
  & M(\keyunisize{\l}, \numkeyu + 1) - \nummembits{\l}{u} \\
  & \ge \bk{\log \binom{\keyunisize{\l}}{\numkeyu + 1} - 1} - \bk{\log \binom{\keyunisize{\l}}{\numkeyu} + 1} - \abs{\log \spilluni{\l}{u} - \log K(\keyunisize{\l}, \numkeyu + 1)} \\
  & \ge \log \binom{\keyunisize{\l}}{\numkeyu + 1} - \log \binom{\keyunisize{\l}}{\numkeyu} - 3
    \;\ge\; w - 1,
\end{align*}
and similarly, $M(\keyunisize{\l}, \numkeyu + 1) - \nummembits{\l}{u} \le w + 10$. Thus, \eqref{eq:few_alloc} holds.

\section{Dictionary via Multi-Level Hashing}
\label{sec:multihash}

In this section, our focus shifts to dynamic dictionaries that exhibit larger redundancy and smaller update time compared to those discussed in the previous section. Up to \cref{sec:rehashing_for_multilevel}, we will design dynamic dictionaries that can store $N$ distinct keys in the universe $[U]$ with $O(R)$ bits redundancy, $O(\log^* N)$ query time, and $O\left(\log^* N + \log \frac{N}{R}\right)$ update time, for any $\frac{N}{\log^{0.1} N} \le R \le \frac{N}{2^{\log^* N}}$; the query time is improved to $O(1)$ in the worst case in \cref{sec:const_query_time}. The upper limit of $R$ already corresponds to the optimal update time $O(\log^* N)$ within this range, so considering larger values of $R$ is unnecessary. It is important to note that when $N$ is sufficiently large, the value of $r = N / R$ exceeds any constant value.

As before, we make the assumption that $U = N^{1 + \alpha}$ for some constant $\alpha > 0$, and $\wordlen = \Omega(\log N)$. We also assume that the dictionary always contains either $N$ keys or $N - 1$ keys, without loss of generality.

\paragraph{Algorithm Overview.}

Let $\enkeyi[1] = \Theta(\log^{9} N)$ be a parameter. We hash keys to $\nbranchi[0] = N / \enkeyi[1]$ buckets which we call the \emph{level-1 buckets}. The expected number of keys hashed into a specific level-1 bucket $u$ equals $\enkeyi[1]$.

For all keys that are mapped to bucket $u$, we further hash them to smaller \emph{level-2 buckets}; the keys in the same level-2 bucket are further hashed to smaller level-3 buckets, and so on. Such procedure is described by a forest: the roots are level-1 buckets and every internal node $u$ at level $\l$ has $\nbranchi$ children at level $\l + 1$. All keys that were hashed into $u$ will further be hashed into a random child of $u$. We use $\enkeyi$ to denote the expected number of keys hashed into a level-$\l$ node: it is clear that $\enkeyi = \enkeyi[\l - 1] / \nbranchi[\l - 1]$. We set $\enkeyi = \Theta\bigbk{(\log^{(\l)} N)^9}$ for internal levels. (Note that the step of hashing all $N$ keys into level-1 buckets is not a part of the tree; we have different technical details for this stage.)

During the hashing process, we carefully guarantee that the tree is not ``too unbalanced''. We define $\devi = \enkeyi^{2/3}$ and hope the number of keys $n$ in a level-$\l$ node is always in $[\enkeyi - \devi, \, \enkeyi + \devi]$. If after we hash all keys in $u$ to its children, some child violates such rule, then we immediately pick another hash function for $u$ and rebuild the subtree rooted at $u$. If some root node already violates this rule, we simply rebuild the whole dictionary. The cost of rehashing is low on average. 
Let $r = N/R$ and $\log^{(k)} N \le r < \log^{(k - 1)} N$. The tree terminates at a level $k$, consisting of leaf nodes, each with expected number of keys $\enkeyi[k] = \Theta(r^9)$. The number of keys is small enough so that a minimap is capable of storing keys in a leaf. The minimap has $O(1)$ query time, $O(\log r)$ insertion/deletion time, and almost no redundancy. Minimaps are the time bottleneck of our algorithm.

However, the number of keys hashed into a leaf is not fixed; it changes with time. This is the main challenge of the algorithm: if we use variable-size data structures for leaf nodes, then it is hard to concatenate them efficiently and succinctly.\footnote{Using \emph{adapters} here is also too slow.} Our solution here is to carefully make each node fixed-size, and assign a fixed piece of memory for each node, as follows:
\begin{itemize}
\item For each leaf node, it contains at least $\lbkeyi[k] \defeq \ceil{\enkeyi[k] - \devi[k]}$ keys. We arbitrarily pick $\lbkeyi[k]$ keys that were hashed into this leaf and use a fixed-size minimap to store them. For the remaining keys that were not stored in the leaf, we send them back to the parent node.
\item Before we process for an inner node $u$ at level $\l$, each child $v$ of $u$ already stored $\lbkeyi[\l + 1]$ keys in its subtree, and all unstored keys are sent to $u$. Recall that the number of keys hashed into the subtree of $u$ is at least $\lbkeyi[\l]$, hence the number of keys sent back to $u$ is at least $\lbkeyi[\l] - \nbranchi[\l] \cdot \lbkeyi[\l + 1]$. We arbitrarily pick this number of keys, store them at node $u$ with a fixed-size hash table, then send other keys to the parent of $u$.
\item Such procedure ends when all unstored keys are sent back to the level-1 bucket, i.e., the root. All keys here are stored in a hash table with enough capacity $\ubkeyi[1] \defeq \floor{\enkeyi[1] + \devi[1]}$. Some of its capacity is wasted since the number of keys is always below the capacity.
\end{itemize}
Below, we add implementation details to complete the construction. We will first introduce the design when everything is ideal, then construct the hash functions that we use, and lastly handle rare bad events via rehashing.

\paragraph{Tree Parameters.} Before we start, we briefly specify the tree parameters. For technical convenience, we let $\nbranchi$ ($\l \in [0, k - 1]$) be always a power of two. We further require $\enkeyi[\l] \defeq \enkeyi[\l - 1] / \nbranchi[\l - 1] \in \left[ (\log^{(\l)} U)^9, \, 2(\log^{(\l)} U)^9 \right)$ for $\l \in [1, k - 1]$; require $\enkeyi[k] \in [r^9, \, 2r^9)$. These requirements uniquely determine the branching factor $\nbranchi$ for all levels. Also note that $r \le \log^{0.1} N$ is guaranteed, meaning that $\enkeyi[2] \le 2 \log^{0.9} N$.

$\enkeyi$ is the expected number of keys hashed into a level-$\l$ node $u$. We only consider the ideal case where the number of keys $n$ is between $[\lbkeyi, \ubkeyi] = [\ceil{\enkeyi - \devi}, \, \floor{\enkeyi + \devi}]$. If not, we say node $u$ \emph{overflows} ($n > \ubkeyi$) or \emph{underflows} ($n < \lbkeyi$). Once overflow or underflow happens, we immediately rebuild some subtree so that it no longer occurs.

\subsection{Hashing with ID-Quotient Separation}

Similar to the previous section, during the hashing steps, we regard keys as ID-quotient pairs. Besides the merit mentioned before that the short ID enables us to use \emph{minimap} for each leaf node, it also allows short pointers to the stored keys. To access a stored key, only the ID of that key needs to be remembered, which saves space.

Specifically, the keys are represented as ID-quotient pairs since they were hashed into tree roots. In a level-$\l$ node $u$, each key is represented by an element in $[2^{\idleni}] \times [\quotunivi]$, where $\idleni \defeq \ceil{4 \log \enkeyi}$ represents the length of ID and $\quotunivi$ represents the range of the quotient.

As the keys received at node $u$ will be hashed into the child nodes, the ID length will become shorter, i.e., of length $\idleni[\l + 1]$; the rest bits will be moved to the quotient. Formally, we label the children of $u$ with integers $v \in [\nbranchi]$. Assume some key $x = (\xid, \, \xquot)$ arrives at $u$. The procedure of hashing $x$ into a child $v$ can be described by a bijective hash function from $\xid$ to $(\chid, \, \yid, \, \xrem)$, where $\chid \in [\nbranchi]$ is the child index that $x$ will enter, $\yid \in [2^{\idleni[\l + 1]}]$ is the ID part of the next-level representation of $x$, and $\xrem$ is the remaining bits that will be merged into the quotient. Then, we concatenate $\xrem$ with $\xquot$, obtaining $\yquot \defeq (\xrem \concat \xquot) \in [\quotunivi[\l + 1]]$ as the next-level quotient. Finally, the pair $(\yid, \, \yquot)$ will be passed to the child $v$.

Similar to the previous section, the hash function $\permUnite : [2^{\idleni}] \to [\nbranchi] \times [2^{\idleni[\l + 1]}] \times [\quotunivi[\l + 1] / \quotunivi]$ used above needs two properties. First, it should evenly distribute all keys to the children, preventing overflowing or underflowing of any child. Second, for each child $\chid \in [\nbranchi]$, all keys assigned to that child should get different $\yid$. The formal construction of the hash function families are defered to \cref{sec:hashfunc_for_multilevel}.

Above we introduced the hashing process at a tree node that distributes the keys to its child nodes. The process of hashing $N$ initial keys to $\nbranchi[0]$ level-1 buckets (tree roots) is similar. We first round up $U$ to a multiple of $\nbranchi[0] \cdot 2^{\idleni[1]}$, then use a hash function $\permUnite : [U] \to [\nbranchi[0]] \times [2^{\idleni[1]}] \times [\quotunivi[1]]$ to map each key $x \in U$ to $(u, \, \yid, \, \yquot)$. Then $(\yid, \, \yquot)$ is passed to the $u$-th level-1 bucket.

\subsection{Storage Structure}

\paragraph{Leaf Nodes.}

Recall that a leaf node $u$ is responsible for storing exactly $\lbkeyi[k]$ keys among all received keys. The selection of keys to store is arbitrary. Let the set of selected keys be $S$. The main structure inside a leaf node is a minimap storing $S$.

Each key $x \in S$ to store is already represented as $(\xid, \, \xquot) \in [2^{\idleni[k]}] \times [\quotunivi[k]]$. To store $S$, we regard $\xid$ as ``keys'' and $\xquot$ as ``values'', then use the minimap to store these key-value pairs.  The set of ID $\Sid = \BK{\xid : x \in S}$ can be encoded within
\[
  \log \binom{2^{\idleni[k]}}{\lbkeyi[k]} \le \log \binom{2^8\cdot r^{36}}{2 r^{9}} \le 2r^9 \log \bk{2^8 \cdot r^{36}} = O(\log^{0.9} N \cdot \log \log N) \le O(w)
\]
bits, where the last inequality holds because $w = \Omega(\log N)$. Thus the ID set $\Sid$ fits in constant words and satisfies the prerequisite of minimap. Applying \cref{lm:minimap_lem} with key universe $[2^{\idleni[k]}]$, value universe $[\quotunivi[k]]$, and redundancy parameter $r = 1$, we get a minimap storing the set $S$.

With minimap, one can query a key in $O(1)$ time, and insert/delete a key in $O(\log r)$ time. The redundancy for the minimap is $O(1)$ bits which comes from the following two parts: (1) \cref{lm:minimap_lem} itself introduces $O(1)$ bits of redundancy; (2) as we want to store the minimap in a fixed-length consecutive space in the memory without using spillover representation, the rounding of spillover introduce $1$ bit of redundancy. Other unstored keys will be returned to the parent node and finally stored at some ancestor. 

\paragraph{Inner Nodes.}

Let $u$ be a level-$\l$ inner node. The number of received keys is guaranteed to be in $[\lbkeyi, \ubkeyi]$. After further hashing these keys into child nodes, each child node stores $\lbkeyi[\l + 1]$ keys in its subtree, and the remaining keys are sent back to $u$. Clearly, there are at least $m \defeq \lbkeyi[\l] - \nbranchi \lbkeyi[\l + 1]$ returned keys. The main task of node $u$ is to store $m$ arbitrarily selected keys that were returned to $u$, so that the subtree of $u$ stores exactly $\lbkeyi$ keys in total.

To finish the task, we maintain a hash table of fixed capacity $m$, called \emph{storage table}. Recall that each key is represented as $(\xid, \, \xquot)$ at this level. The storage table takes $\xid$ as its key and $\xquot$ as its value. We choose the hash table implementation from \cite{liu22} with $O(m \log \log 2^{\idleni}) = O(m \log \idleni)$ redundancy and $O(1)$ query/update time. Note that the storage table is always full and takes a fixed space, so we put it at a static location in the memory.

There will always be unstored keys; we again return them to the parent. The only exception is the root nodes. At root nodes, we slightly change the configuration, letting $m = \ubkeyi[1] - \nbranchi[1] \lbkeyi[2]$. Then, as long as the root node does not overflow, all keys can be successfully stored in the storage table. In this case, a portion of capacity is wasted, but the produced redundancy is still acceptable.

\paragraph{Waiting Lists.}

The above construction is already enough to encode the set of keys. However, to support quick deletion, we need one more auxiliary structure. Let $u$ be a level-$\l$ node where $\l \ge 2$. After storing a fixed number of keys in the storage table, other unstored keys are returned to the parent of $u$. Imagine that we need to remove some key from $u$'s storage table. To make the storage table full again, we need to pull down some key $x$ from $u$'s ancestors and then insert $x$ to $u$'s storage table. \emph{Waiting list} is designed for quickly finding such $x$.

For every key $x = (\xid, \, \xquot)$ that is not stored in the subtree of $u$, i.e., has been sent back to $u$'s parent, we store its current-level ID $\xid$ and last-level ID $\xid^{[\l - 1]}$ in the waiting list. Its tail part is not stored here. The waiting list has a fixed capacity $\ubkeyi - \lbkeyi = \Theta(\devi)$. As long as node $u$ does not overflow, this capacity is enough. (The whole waiting list is auxiliary redundant information, so underutilizing its capacity is not a problem.) When we need to find a key $x$ stored at ancestors of $u$, we just pick an arbitrary $\xid^{[\l - 1]}$ from the waiting list of $u$, then query $x$ at $u$'s parent. This is doable because $\xid^{[\l - 1]}$ serves as the unique identifier of $x$ at $u$'s parent.

To implement the waiting list, we straightforwardly list all $(\xid, \, \xid^{[\l - 1]})$ as bit strings in the memory, with a 1-bit separation between adjacent two terms which marks the end point of the list. The waiting list can hold at most $\ubkeyi - \lbkeyi \le 2 \devi$ keys, so the space usage is at most
\[2 \devi (1 + \idleni + \idleni[\l - 1]) \le 6 \devi[2] \cdot \idleni[1] = 6 \enkeyi[2]^{2/3} \cdot \ceil{4 \log \enkeyi[1]} < O(\log N) = O(\wordlen),\]
where the second inequality holds because $\enkeyi[2] \le \log^{0.9} N$. Hence, the whole waiting list can fit in $O(1)$ words, meaning that all operations on the waiting list can be done in $O(1)$ time via bit operations.

All leaf nodes and inner nodes except the roots maintain their own waiting lists. Root nodes do not need waiting lists because their storage tables are large enough to contain all keys.

\paragraph{Redundancy Analysis.}

As introduced above, the dictionary is composed of a range of fixed-size structures:
\begin{itemize}
\item a minimap for each leaf node;
\item a storage table for each non-leaf node;
\item a waiting list for each non-root node.
\end{itemize}
(The space usage of hash functions are not counted here.) Next, we carefully list the redundant information in the dictionary:
\begin{itemize}
\item Each minimap produces $O(1)$ redundancy. Multiplied by the number of leaf nodes $N / \enkeyi[k]$, there is a redundancy of $O(N / \enkeyi[k]) = O(N / r^9)$.
\item A waiting list in level $\l$ occupies $O(\devi \cdot \idleni[\l - 1]) = O(\enkeyi^{2/3} \log \enkeyi[\l - 1]) \le O(\enkeyi^{2/3} \cdot \enkeyi^{1/9}) = O(\enkeyi^{7/9})$ bits of memory. These bits are all redundant. Multiplied by the number of level-$\l$ nodes $N / \enkeyi$, it is a redundancy of $O(N / \enkeyi^{2/9})$ bits. It reaches maximum when $\l = k$, i.e., at the leaf nodes, where the total redundancy is $O(N / r^2)$.
\item According to the hash table implementation, a storage table in level $\l$ produces $O(m \log \log 2^{\idleni}) = O(m \log \log \enkeyi)$ redundancy, where $m = \lbkeyi - \nbranchi \lbkeyi[\l + 1]$ for non-root level $\l$ and $m = \ubkeyi[1] - \nbranchi[1] \lbkeyi[2]$ for the root level. As each child can only contribute $2 \devi[\l + 1]$ unstored keys, $m = O(\nbranchi \devi[\l + 1])$. The total redundancy over all level-$\l$ nodes is
  \[
    \frac{N}{\enkeyi} \cdot O(\nbranchi \devi[\l + 1] \log \log \enkeyi)
    \le O\bk{\frac{N}{\enkeyi[\l + 1]^{1/3}} \log \enkeyi[\l + 1]}.
  \]
  This reaches maximum when $\l + 1 = k$, i.e., at the level above leaves. The corresponding redundancy is $O(N/r^3 \cdot \log r)$.
\item Each root's storage table wastes $O(\devi[1])$ capacity; wasting one capacity results in $O(\log U)$ redundancy. The total redundancy of this part is $O(N / \enkeyi[1]) \cdot O(\devi[1] \log U) = O(N / \log^{2} N)$.
\item Key selection. When we select $\enkeyi$ keys to store at $u$, ``which keys are selected'' is an implicit piece of redundant information that was stored in the dictionary. Its entropy is at most $\log \binom{\ubkeyi}{\lbkeyi} = O(\devi \log \enkeyi)$. It is dominated by the redundancy of the waiting list on the same node.
\item Lookup table used by minimap. Recall that we use minimap to maintain $n' \defeq \lbkeyi[k] = O(r^9) = O(\log^{0.9} N)$ key-value pairs from a key universe $U' \defeq 2^{\idleni[k]} = \poly \log N$, allowing $O(1)$ bits redundancy. By \cref{lm:minimap_lem}, the lookup table size for this minimap is $O(U' n' \binom{U'}{n'} \cdot (2n')^{n'}) = (\poly \log N)^{O(\log^{0.9} N)} \ll O(\sqrt{N})$ words. 
\end{itemize}
The largest redundancy appears at the second item: the waiting lists of leaf nodes. Summing up the redundancy from all categories, the overall redundancy of the dictionary is $O(N / r^2) \le O(R)$ which meets the requirement.

\subsection{Query, Insertion, and Deletion}

Provided the storage structure of the dictionary, it is straightforward to perform operations efficiently, including queries, insertions, and deletions.

In the beginning of each opeartion, we need to simulate the hashing process on the input key $x$, trace all the nodes which $x$ was hashed to, and obtain a path from a root node to a leaf node. During this process, we also record the ID-quotient pairs $(\xid^{[\l]}, \, \xquot^{[\l]})$ in each node along the path. The next step varies depending on the operation:
\begin{itemize}
\item Query. For the level-$\l$ node on the path, we query $\xid^{[\l]}$ in its storage table (if it is an inner node) or minimap (if it is a leaf). If at any node the query result is ``exist'', and the tail part in the query result equals $\xquot^{[\l]}$, then we return ``exist''.\footnote{If $\xid^{[\l]}$ is found but the tail part is not $\xquot^{[\l]}$, we immediately know $x$ is not in the dictionary, by the uniqueness of the ID.} Otherwise, as $x$ is not stored in any node along the path, we report that $x$ is not in the dictionary.
\item Insertion. Since all storage tables and minimaps are full except for the root, we insert the key directly into the root's storage table. We also add $x$ to the waiting lists of all non-root nodes along the path (strictly speaking, add $(\xid^{[\l]}, \, \xid^{[\l - 1]})$ to the waiting list of level $\l$).
\end{itemize}

The deletion procedure is shown in \cref{alg:delete_large_redun}. Assume $x$ was stored at node $u$ before the deletion. The motivation is to try to delete $x$ from the storage table or minimap of $u$. To make it full again, we access the waiting list to find another key $y$ that was stored at an ancestor of $u$. Then we recursively delete $y$ from the ancestor while inserting $y$ into $u$'s storage table or minimap. We also carefully maintain the waiting lists along this procedure. See \cref{alg:delete_large_redun} for details.

\begin{algorithm}[ht]
  \caption{Deletion}{\label{alg:delete_large_redun}}
  \DontPrintSemicolon
  \SetKwProg{Fn}{Function}{:}{}
  
  \SetKwFunction{fDel}{DeleteAndGetQuotient}
  \Fn(\Comment{$u$ is the level-$\l$ node along the path.}){\fDel{$u$, $\xid^{[\l]}$}} {
    \uIf{$u$ is the root} {
      Delete $\xid^{[\l]}$ from the storage table while getting its quotient $\xquot^{[\l]}$. \;
      \Return $\xquot^{[\l]}$
    } \uElseIf{$x$ is stored in the storage table or minimap of $u$} {
      Delete $x$ from the storage table or minimap while getting its quotient $\xquot^{[\l]}$. \;
      Delete an arbitrary key $y$ from the waiting list of $u$ while obtaining $\yid^{[\l - 1]}$. \;
      $\yquot^{[\l - 1]} \gets$ \fDel{$\textup{parent}(u)$, $\yid^{[\l - 1]}$} \label{line:recurse_del} \;
      Insert $y$ to the storage table or minimap of $u$. \;
      \Return $\xquot^{[\l]}$.
    } \Else(\Comment{$x$ is in the waiting list of $u$.}) {
      Delete $x$ from the waiting list while getting $\xid^{[\l - 1]}$. \;
      $\xquot^{[\l - 1]} \gets$ \fDel{$\textup{parent}(u)$, $\xid^{[\l - 1]}$} \;
      \Return $\xquot^{[\l]}$ obtained from $\xquot^{[\l - 1]}$ and $\xid^{[\l - 1]}$
    }
  }

  Let $u$ be the leaf node along the path. \;
  \fDel{$u$, $\xid^{[k]}$}
\end{algorithm}

Note that when we read an entry $y$ from the waiting list, we do not know the quotient of $y$, which is required to insert $y$ to the current node $u$. The solution is that we first recursively remove $y$ from the ancestors and simultaneously learn its quotient (\cref{line:recurse_del}).

All above operations have expected running time $O(\logstar N + \log r) = O(\logstar N + \log (N / R))$. First, the tree has $O(\logstar N)$ levels, while updating a storage table or a waiting list takes $O(1)$ expected time. Also, every node on the path will only be visited constant times, so the time on operating inner nodes is $O(\logstar N)$. Second, the minimap on leaf nodes takes $O(\log r)$ update time; $O(1)$ minimap updates are required in every single dictionary operation. Combining them together leads to the desired time complexity. Furthermore, query operations only take $O(\logstar N)$ time because they do not need minimap updates (minimap queries are done in constant time). The query time will be improved to $O(1)$ in \cref{sec:const_query_time}.

\subsection{Construct the Hash Function}
\label{sec:hashfunc_for_multilevel}

So far, we have introduced the basic construction of our algorithm except for one detail -- the hash function $\permUnite$ that assigns child indices and next-level IDs for the keys. For each level $\l$, we will specify a function family $\hashfami$, and sample a function $\permUnite$ from $\hashfami$ uniformly at random for each level-$\l$ node $u$. Recall that we have the following two requirements for $\permUnite$:
\begin{enumerate}[label=(\arabic*)]
\item\label{enum:no_overflow} $\permUnite$ distributes all keys received by $u$ evenly among its children so that the number of keys assigned to any child $v$ is between $[\lbkeyi[\l + 1], \ubkeyi[\l + 1]]$.
\item\label{enum:unique_id} For a fixed child $v$ of $u$, all keys assigned to $v$ have different level-$(\l + 1)$ IDs.
\end{enumerate}
We require these two properties to hold with high probability. If any of them is violated, we immediately resample\footnote{The new function may depend on the previous one; it is not necessarily uniformly random from $\hashfami$.} a new $\permUnite$. Besides, we also need $\permUnite$ to have small space and evaluation time so that it will not be the bottleneck of our dictionary.

\paragraph{Challenges.} Before we give the concrete construction of our hash function, let us first see why picking a uniform random permutation as $\permUnite$ does not work. Formally, we first pick a uniform random permutation $\permUnite : [2^{\idleni}] \to [\nbranchi] \times [2^{\idleni[\l + 1]}] \times [\quotunivi[\l + 1] / \quotunivi]$; as long as any of the two requirements does not meet, we repeatedly pick another one at random. The problem is that Property~\ref{enum:unique_id} above has too low probability to be satisfied. For every child $v$ of $u$, there are about $\enkeyi[\l + 1]$ keys assigned to it. When every key is assigned a random ID in a universe $2^{\idleni[\l + 1]} \approx \enkeyi[\l + 1]^4$, the probability of existing an ID collision at $v$ is about $\Theta(1 / \enkeyi[\l + 1]^2)$. However, we have $\nbranchi = \enkeyi / \enkeyi[\l + 1] \gg \enkeyi[\l + 1]^2$ children, so there is a child violating Property~\ref{enum:unique_id} with high probability.

In this example, when some child $v$ of $u$ violates the uniqueness of IDs, we resample the whole function $\permUnite$ while ignoring the success of other children. This is indeed unnecessary -- by carefully construct the function family $\hashfami$, we are able to only reassign level-$(\l + 1)$ IDs in $v$ while keeping the function values in other children, which is the motivation of our two-stage construction below.

\paragraph{Construction.} We let the hash function $\permUnite$ be a composition of two subfunctions $\permCh$ and $\permId$. $\permCh$, sampled from family $\chidfami$ at random, extracts several bits from the input key's $\xid$ to be the child index $v \in [\nbranchi]$; then, for each child $v$, an individual bijection $\permId$ (sampled from family $\idfami$) extracts several bits again to be the ID inside that child. Formally,
\[
  \begin{array}{clcl}
    & \permCh : [2^{\idleni}] \to [\nbranchi] \times [2^{\idleni} / \nbranchi] & \; \textup{maps} \;
    & \xid \mapsto (v, \, \yprim), \\
    \forall v \in [\nbranchi], & \permId : [2^{\idleni} / \nbranchi] \to [2^{\idleni[\l + 1]}] \times [2^{\idleni - \idleni[\l + 1]} / \nbranchi] & \; \textup{maps} \; & \yprim \mapsto (\yid, \, \xrem).
  \end{array}
\]
Here, the intermediate variable $\yprim$ is called the \emph{primitive ID}. Similar to $\yid$, $\yprim$ is unique among all keys mapped to the child $v$, but it is much longer. (Its uniqueness comes from the fact that $\permCh$ is a bijection.) Note that different children $v$ have different $\permId$, which allows us to resample for a single child when an ID collision occur. After decomposing $\permUnite$ into two subfunctions, Property~\ref{enum:no_overflow} becomes a requirement of $\permCh \in \chidfami$; Property~\ref{enum:unique_id} becomes a requirement of $\permId \in \idfami$.

The existence of function families $\chidfami$, $\idfami$ is summarized as the following lemmas.

\begin{restatable}{lemma}{HashFamilyOne}
  \label{lem:hashfamily_1}
  Let $s, L$ be integers where $1 \le s \le (1 - \Omega(1)) L$. There is a hash function family $\hashfam$ in which every member $h \in \hashfam$ is a bijection $h : \BK{0, 1}^{L} \to \BK{0, 1}^{L}$, satisfying:
  \begin{enumerate}[label=\textup{(\alph*)}]
  \item For any $h \in \hashfam$ and any input $x \in \BK{0, 1}^L$, $h(x)$ and $h^{-1}(x)$ can be evaluated in $O(1)$ time.
  \item It takes $O(2^{\eps L})$ bits to store an $h \in \hashfam$, where $\eps$ is a constant given in advance which does not depend on $s$ and $L$.
  \item\label{enum:no_overflow_in_lemma} For $n \ge 2^{s} \cdot s^4$ different inputs $x_1, \ldots, x_n$, if we divide $h(x_1), \ldots, h(x_n)$ into equivalent classes according to the first $s$ bits of $h(x_i)$, then with probability $\ge 1 - \frac{1}{4n^2}$, the number of elements in any equivalent class is between
    $\Bk{ \frac{n}{2^{s}} - \frac{1}{5} \bk{\frac{n}{2^{s}}}^{2/3}, \, \frac{n}{2^s} + \frac{1}{5} \bk{\frac{n}{2^s}}^{2/3}}$.
  \end{enumerate}
\end{restatable}

We defer the proof of \cref{lem:hashfamily_1} to \cref{sec:hashfunc} and first see how we use it to construct $\chidfami$. In fact, taking the parameters $\eps = 1/100$, $L = \idleni = \ceil{4 \log \enkeyi}$, and $s = \log \nbranchi$ (the premise $s \le (1 - \Omega(1)) L$ holds as $2s \le 2 \log \enkeyi < 4 \log \enkeyi \le L$), we get a function family that supports $O(1)$ time evaluation, and takes space $O\bigbk{2^{\frac{1}{100} \idleni}} \le O\bigbk{2^{\frac{1}{100} \ceil{4 \log \enkeyi}}} \ll O\bigbk{\enkeyi^{2/3}}$. Recall that a level-$\l$ node maintains a waiting list of space $\Theta\bigbk{\enkeyi^{7/9}}$, the space cost of storing the hash function $\permCh \in \chidfami$ is not the bottleneck.\footnote{Although the root nodes do not maintain waiting lists, one can still verify that the space cost of storing $\permCh$ is dominated by the redundancy of the node's storage table.}

Substituting $n \in [\lbkeyi, \ubkeyi]$ in Property~\ref{enum:no_overflow_in_lemma} (the premise $n \ge 2^s \cdot s^4$ holds as $n / 2^s = \Theta(\enkeyi[\l + 1]) \ge \Omega(\log^{9} \enkeyi) \gg s^4$), the interval given in \ref{enum:no_overflow_in_lemma} becomes
\[\Bk{\frac{n}{\nbranchi} - \frac{1}{5}\bk{\frac{n}{\nbranchi}}^{2/3}, \; \frac{n}{\nbranchi} + \frac{1}{5}\bk{\frac{n}{\nbranchi}}^{2/3}}. \label{eq:interval_in_c}\]
It is easy to verify that this interval is contained\footnote{More strictly, the integers in this interval are contained in $[\lbkeyi[\l + 1], \, \ubkeyi[\l + 1]]$. It is enough because we are bounding the number of elements in an equivalent class, which is always an integer.} in $[\lbkeyi[\l + 1], \, \ubkeyi[\l + 1]]$. For example, the upper limit
\begin{align*}
  \frac{n}{\nbranchi} + \frac{1}{5} \bk{\frac{n}{\nbranchi}}^{2/3}
  &\le \ubkeyi \cdot \frac{\enkeyi[\l + 1]}{\enkeyi} + \frac{1}{5} \bk{\ubkeyi \cdot \frac{\enkeyi[\l + 1]}{\enkeyi}}^{2/3} \\
  &\le \bk{1 + \frac{\devi}{\enkeyi}} \enkeyi[\l + 1] + \frac{1}{5} \bk{\bk{1 + \frac{\devi}{\enkeyi}} \enkeyi[\l + 1]}^{2/3} \\
  &\le \enkeyi[\l + 1] + \frac{\enkeyi[\l + 1]}{\enkeyi^{1/3}} + \frac{1}{4} \enkeyi[\l + 1]^{2/3} \\
  &\le \enkeyi[\l + 1] + \enkeyi[\l + 1]^{2/3},
\end{align*}
whose integral part equals $\ubkeyi[\l + 1]$. Here, the third inequality holds because $\bk{1 + \devi / \enkeyi}^{2/3} \le 5/4$ as long as $\enkeyi$ is not too small; the last inequality holds because $\enkeyi^{1/3} \ge (3 \enkeyi[\l + 1])^{1/3} > \frac{4}{3} \enkeyi[\l + 1]^{1/3}$. It is similar for the lower limit.

From \ref{enum:no_overflow_in_lemma}, the probability of any child overflowing or underflowing is at most $1/4n^2 \le 1/\enkeyi^2$ (we have $n \ge \enkeyi/2$ as $\enkeyi$ is larger than a constant).

In summary, by applying \cref{lem:hashfamily_1} with proper parameters, we get the function family $\chidfami$ such that with probability at least $1 - 1 / \enkeyi^2$, no child of $u$ overflows or underflows. Moreover, any function $\permCh \in \chidfami$ and its inverse can be evaluated in constant time; its storage space is also not the bottleneck of our algorithm.

Compared to $\chidfami$, $\idfami$ has a stricter space constraint, because each child $v$ of $u$ needs to store a copy of $\permId$ independently. Correspondingly, its independence constraint is less strict.\footnote{We only need 2-wise independence for $\permId$. See the appendix for details.} Similar to \cref{lem:hashfamily_1}, we will prove the following lemma in \cref{sec:hashfunc}, which gives the construction of $\idfami$:

\begin{restatable}{lemma}{HashFamilyTwo}
  \label{lem:hashfamily_2}
  Let $s, L$ be integers where $0 < s < L$. There is a hash function family $\hashfam$ in which every member $h \in \hashfam$ is a bijection $h : \BK{0, 1}^L \to \BK{0, 1}^L$, satisfying:
  \begin{enumerate}[label=\textup{(\alph*)}]
  \item For any $h \in \hashfam$ and any input $x \in \BK{0, 1}^L$, $h(x)$ and $h^{-1}(x)$ can be evaluated in constant time.
  \item It takes $O(L)$ bits to store an $h \in \hashfam$.
  \item\label{enum:no_collision} For any integer $n$ and $n$ fixed distinct inputs $x_1, \ldots, x_n$, with probability at least $1 - n^2 / 2^s$, the first $s$ bits of $h(x_1), \ldots, h(x_n)$ are pairwise distinct.
  \end{enumerate}
\end{restatable}

Applying \cref{lem:hashfamily_2} with $L = \idleni - \log \nbranchi$ and $s = \idleni[\l + 1]$, we get a function family $\idfami$ which supports constant time evaluation and $O(\idleni)$ storage space. We first check that the space usage of $\permId$ does not become the bottleneck of our dictionary: For a level-$\l$ node $u$ and its child $v$, we use $O(\idleni)$ bits of space to store $\permId$ at node $v$, while $v$ already has a larger redundant waiting list of space $\Omega(\idleni \cdot \devi[\l + 1])$.

In Property~\ref{enum:no_collision}, let $n$ be any integer in $[\lbkeyi[\l + 1], \ubkeyi[\l + 1]]$, then the probability of having a collision is at most
\[
  \frac{n^2}{2^{\idleni[\l + 1]}} \le \frac{\ubkeyi[\l + 1]^2}{\enkeyi[\l + 1]^4} \le \frac{2}{\enkeyi[\l + 1]^2},
\]
where the last inequality holds as long as $\enkeyi[\l + 1]$ is larger than a constant. Hence, for level-$\l$ node $u$ and its child $v$, such hash function $\permId$ is able to guarantee the uniqueness of the level-$(\l + 1)$ IDs in $v$ with probability $1 - 2 / \enkeyi[\l + 1]^2$, which meets our requirement. So far, by combining $\permCh \in \chidfami$ and $\permId \in \idfami$, we conclude the construction of the hash function $\permUnite$.

\paragraph{Hash functions in the top level.} Although the above discussion is all about hash functions on the tree nodes, the bijective hash function that we use to hash initial keys into level-1 buckets (i.e., tree roots) are essentially the same. If $U$ is a power of two, then we just regard the whole key $x \in [U]$ as an ID (while the quotient is empty), and apply \cref{lem:hashfamily_1,lem:hashfamily_2} as above to construct corresponding hash functions. If $U$ is not a power of two, we only need to slightly generalize the construction in \cref{lem:hashfamily_1,lem:hashfamily_2}. The details are deferred to \cref{sec:hashfunc}.

\subsection{Rehashing}
\label{sec:rehashing_for_multilevel}

When the hash function at some node does not meet the requirements, we need to resample a hash function and then reconstruct part of the data structure. This procedure is called \emph{rehashing}. In this subsection, we show that the rehashing frequency is low enough so that the expected time cost of rehashing will not become the bottleneck of our algorithm.

\newcommand{\famiGL}[1][\ell]{\mathcal{G}_{#1}}

We start by introducing the behavior of rehashing. Denote by $\permchu$ the hash function $\permch \in \chidfami$ for node $u$, and recall that $\permId \in \idfami$ is stored at the child $v$. They distribute the keys received by $u$ to the children and assign IDs of level $\l + 1$ to these keys. When $\permchu$ lets a child overflow or underflow, we resample $\permchu$ uniformly at random from $\chidfami$, and reconstruct the whole subtree rooted at $u$. For any child $v$, when $\permId$ leads to an ID collision in $v$, we resample $\permId$ randomly from $\idfami$, and then reconstruct the subtree rooted at $v$.

Notice the similarity between $\permchu$ and $\permIdu$ -- any failure among them leads to the reconstruction of the same subtree rooted at $u$; also, they have similar failure probabilities ($1 / \enkeyi^2$ vs $2 / \enkeyi^2$). For simplicity, we consider these two subfunctions as a whole: let $g_u = (\permchu, \, \permIdu) \in \famiGL = \chidfami \times \idfami[\l - 1]$ be the function stored at $u$. Note that this grouping is different from previous subsections (where $\permchu$ is grouped with the children's $\permId[v]$). When one of $\permchu$, $\permIdu$ fails to satisfy the requirement, we say $g_u$ \emph{fails}, in which case we resample the whole $g_u$ instead of only the failed part. From the previous subsection, we know that a random $g_u$ from $\famiGL$ fails with probability at most $3 / \enkeyi^2$ for a fixed set of input keys.

The recursive process of reconstructing a subtree is described below in \cref{alg:rehash}.

\begin{algorithm}[ht]
  \caption{Rehashing}{\label{alg:rehash}}
  \DontPrintSemicolon
  \SetKwProg{Fn}{Function}{:}{}

  \SetKwFunction{rehash}{Rehash}
  \SetKwFunction{leadrehash}{LeadRehashing}

  \Fn{\rehash{$u$}} {
    Sample $g_u \in \famiGL$ uniformly at random.\;
    \If{$g_u$ fails} {
      \rehash{$u$}\;
      \Return\;
    }
    Distribute the keys to children according to $g_u$.\;
    \ForEach{child $v$ of $u$} {
      \rehash{$v$}\;
    }
    Fill the storage table or minimap at $u$.\;
    Fill the waiting list of $u$ if $u$ is not the root.\;
  }
\end{algorithm}

Rehashing the subtree rooted at a level-$\l$ node $u$ takes $O(\enkeyi (k - \l + 1)) \le O(\enkeyi \logstar \enkeyi)$ time: the function \rehash takes $O(\enkeyi[\l'])$ expected time processing every level-$\l'$ node; summing over all descendants of $u$ leads to this time complexity. It is possible that the resampled $g_u$ still fails, in which case we immediately rehash $u$ once more.

When the hash function $g$ fails at an ancestor $p$ of $u$, the \rehash{$u$} is also called in the recursion. In this case, we say $p$ \emph{causes} the rehashing. However, we still count the time cost of this rehashing on every descendant of $p$. Once \rehash{$u$} is called, we count $O(\enkeyi)$ time cost on node $u$, regardless whether $u$ causes the rehashing or not. Then, the total time usage of rehashing equals the sum of time costs over all nodes. Based on this idea, we prove the following lemma:

\newcommand{\subseq}{A}
\newcommand{\nrhash}{m}
\newcommand{\nrhashi}[1][\l]{a_{#1}}

\begin{lemma}
  \label{lem:rehashing_time}
  Assume there are $M$ operations (insertion/deletion) after the initialization of the dictionary with capacity $N$. The expected time of rehashing in this procedure is at most $O((N + M) \logstar N)$.
\end{lemma}
\begin{proof}
  Let $u$ be a level-$\l$ node. We consider every maximal consecutive subsequence $\subseq$ of operations such that no ancestor of $u$ is rehashed. We also ignore the operations in $\subseq$ where the key to insert or delete is not hashed into $u$. Throughout $\subseq$, $u$ may be rehashed multiple times (which are all caused by $u$), but the input of $g_u$ (i.e., the primitive ID $\xprim$ of each key $x$) is unchanged except for the keys being deleted or inserted. Below, we first derive an upper bound for such a maximal subsequence, then sum over all nodes $u$ and maximal subsequences. When doing such analysis, we fix the hash functions $g$ and the rehashing history of $u$'s ancestors, and only consider the randomness of sampling $g_u$ when rehashing.

  \begin{claim}
    \label{clm:rehashing}
    Assume $u$ is a level-$\l$ node and $\subseq$ is a maximal consecutive subsequence of operations where $u$'s ancestors are not rehashed, where at the $t$-th operation in $\subseq$ there is a rehashing caused by $u$. Then, with probability at least $1/2$ under the randomness of resampled $g_u$, there will not be rehashing caused by $u$ in the next $\midceil{\enkeyi^2 / 10}$ operations in $\subseq$.
  \end{claim}
  \begin{proof}
    We define $S_{i}$ as the set of primitive IDs of the keys hashed into $u$ after the $(t + i)$-th operation in $\subseq$, for $0 \le i \le \ceil{\enkeyi^2 / 10}$. For a random $g_u \in \famiGL$, the probability that $g_u$ fails on $S_i$ is at most $3 / \enkeyi^2$, provided $|S_i| \in [\lbkeyi, \, \ubkeyi]$, which is guaranteed since there is no rehashing at $u$'s ancestors. Applying the union bound over $i = 0, 1, \ldots, \midceil{\enkeyi^2 / 10}$, we know that $S_0, \ldots, S_{\midceil{\enkeyi^2 / 10}}$ do not let $g_u$ fail with probability at least $1 - \bk{\frac{\enkeyi^2}{10} + 1}\bk{\frac{3}{\enkeyi^2}} \ge \frac{1}{2}$. In the latter case, $u$ will not be rehashed in the next $\midceil{\enkeyi^2/10}$ operations.
  \end{proof}

  As a direct implication of \cref{clm:rehashing}, if the subsequence $\subseq$ consists of $T$ operations, then the expected number of rehashings of $u$ is $O(T/\enkeyi^2 + 1)$. The term $O(1)$ refers to the inevitable rehashing of $u$ at the beginning of $\subseq$, which is caused by an ancestor of $u$.

  Next, we take the summation of the time costs over all nodes. Denote by $\nrhashi$ the total number of rehashings of level-$\l$ nodes. We write the recurrence
  \begin{align*}
    \nrhashi &\le \nrhashi[\l - 1] \cdot \nbranchi[\l - 1] + \frac{M}{\enkeyi^2}, \qquad (\l > 1)
  \end{align*}
  where the first term refers to the number of rehashings that are recursively called from the parents; the second term is the summation of $O(T / \enkeyi^2)$ terms, as the sum of $T$ of all level-$\l$ nodes equals $M$ since they form a partition of the operation sequence. (The notation $O$ is omitted for cleanliness.) Similarly, the initial value is
  $\nrhashi[1] \le \frac{N}{\enkeyi[1]} + \frac{M}{\enkeyi[1]^2}$.
  Solving this recurrence, we have
  \begin{align*}
    \nrhashi \le \frac{N}{\enkeyi} + \frac{M}{\enkeyi} \sum_{i=1}^{\l} \frac{1}{\enkeyi[i]} \le \frac{N}{\enkeyi} + \frac{M}{\enkeyi}.
  \end{align*}
  Each rehashing of level $\l$ takes $O(\enkeyi)$ time, so the total time usage of all rehashings is
  \[
    \sum_{\l = 1}^k \nrhashi \cdot \enkeyi
    \le k(N + M)
    \le (N + M) \logstar N.
    \qedhere
  \]
\end{proof}

\subsection{Achieving Constant-Time Queries}
\label{sec:const_query_time}

The dictionary we described so far has $\Theta(\log^* N)$ query time, which is due to simulating the hashing process on the input key $x$ along the path of length $O(\log^* N)$ from a root to a leaf. We can improve it as follows: After only $O(1)$ hashing steps at the beginning, the IDs of the keys become very short ($\ll \log \log N$), and the memory usage to store the set of IDs is far less than a word. If we rearrange the order of information in the memory such that the information of ID set occurs consecutively, it is possible to use lookup tables to complete the remaining steps of the query. This will improve the query time to $O(1)$ in the worst case.

Specifically, each level-$\l$ node $u$ ($\l \ge 3$) is responsible for storing a set of $\poly \log^{(\l)} N \ll \log \log N$ IDs from a universe $2^{\idleni[\l]} \ll \log \log N$, and an associated quotient of $O(w)$ bits for each ID. Our first observation is that we can store the ID and quotient parts separately. On leaf nodes, we used minimaps to store the information, which consists of the ID set and a fixed number of quotient slots. Even if we store these two parts in different regions of the memory, the minimap still works.\footnote{Technically, we extracted a spillover from every quotient before storing it in the memory, which is encoded and stored along with the key set (in the former part).} For internal nodes, we used succinct hash tables from~\cite{liu22} to store a set of ID-quotient pairs. Their hash tables also store quotients in fixed slots, and thus work even if we store the metadata (including the ID set) separately from the quotient slots. We also store the hash functions at a node as part of the metadata.

Then, we store the metadata at all nodes in the depth-first pre-order in the memory, so that the metadata within a subtree are arranged consecutively. For the level-$\l$ node $u$ ($\l \ge 3$), the total space usage of the metadata for all its descendants does not exceed $\poly \log \log N \ll \log^{0.1} N$, which fits in a machine word so that we can access in $O(1)$ word-accesses. Furthermore, these metadata are enough for us to perform the remaining steps of the query, locating the quotient slot we need to access (if there is one). A lookup table of $2^{O(\log^{0.1} N)} = N^{o(1)}$ words is sufficient to support it in worst-case constant time.

In sum, by rearranging the order of storing information in the memory, and with the help of a lookup table, we can improve the query time of the dictionary to constant in the worst case. This covers the regime of $R \ge N / \log^{0.1} N$ in \cref{thm:main}.

\bibliographystyle{alpha}
\bibliography{reference.bib}

\appendix

\section{Proof of Random Swap Lemma}
\label{sec:randomswap}
\newcommand{\Lr}{L_{\textup{r}}} 
\newcommand{\Lk}[1][k]{L_{#1}}
\newcommand{\sk}[1][k]{s_{#1}}
\newcommand{\mtemp}{m_{\text{temp}}}
\newcommand{\lentail}{T}
\newcommand{\store}[1]{\texttt{Store}\bk{#1}}
\newcommand{\malloc}{m'}
\newcommand{\skplus}{\bk{\bk{\sk + 1} \bmod{\Lk}}}

In this section, we provide the proof of the following random swap lemma introduced in \cref{sec:dabtree}:

\randomswap*

\begin{proof}
  For the given VM of $M$ bits, we use $L \defeq \floor{M / w}$ to denote the number of words in it, $\lentail \defeq M - wL$ to denote the number of bits in its tail, and $m$ to denote its whole memory string. Below we illustrate how to store this VM with memory string $m$ into a WVM.

  As mentioned in \cref{sec:wvm}, we swap the tail with a random word to amortize the slow access to the tail.
  The first step is to cut the tail into $T$ individual bits and store them into the bit tape of the WVM. When the tail is ``too short'' and the number of words is not ``very few'' (explained later), we additionally cut the last word into $w$ bits as well. Next, we swap the contents cut into bits with one or two preceding words at a random location. Finally, the bit tape contains the new individual bits (after swap), and we store the rest of the words into the word tape. See \cref{alg:random_swap}.

  \begin{algorithm}[htbp]
    \caption{Storing VM into WVM}
    \label{alg:random_swap}
    \DontPrintSemicolon
    
    \SetKwFunction{fencode}{Store}
    \SetKwProg{Fn}{Function}{:}{}
    
    \Fn(\Comment{Store a VM with memory string $m$ into a WVM}){\fencode{$m$}} {
      Sample $\Lr$ from $\mleft[\frac{w}{20},\frac{w}{11}\mright)$ unformly at random\; \label{line:sample_start} 
      \ForEach{$i \in [\log N]$}{
        Sample $\Lk[i]$ from $[2^{i}, 2^{i+1})$ uniformly at random\;
        Sample $\sk[i]$ from $[1, \Lk[i]]$ uniformly at random\;
      }
      Sample $a$ from $[0,w-2]$ uniformly at random\; \label{line:sample_end}
      $k \gets$ the unique index such that $\Lk < L \le \Lk[k+1]$\;
      \uIf{$L \le \Lr$ \emph{or} $\lentail > a$} { \label{line:case_1_2}
        Swap the tail of $m$ with the first $T$ bits of the $\sk$-th word of $m$\; \label{line:swap_process_1}
        $\mtemp \gets$ the resulting memory string after swap\;
        \Return a WVM with its word tape storing all $L$ words of $\mtemp$ and its bit tape storing the tail of $\mtemp$\;
      } \Else { \label{line:case_3}
        Swap the last word of $m$ with the $\sk$-th word of $m$\; \label{line:swap_process_2}
        Swap the tail of $m$ with the first $T$ bits of the $\skplus$-th word of $m$\; \label{line:swap_process_3}
        $\mtemp \gets$ the resulting memory string after swapp\;
        \Return a WVM with it word tape storing the first $L-1$ words of $\mtemp$ and the bit tape storing the last word plus the tail of $\mtemp$\;
      }
    }
  \end{algorithm}

  \bigskip
  Note that the algorithm exhibits different behaviors depending on the length of $m$. There are three cases:
  \begin{itemize}
  \item The \emph{Few-Word Case}: the number of words $L \le \Lr$.
  \item The \emph{Short-Tail Case}: the length of the tail $T \le a$ while $L > \Lr$, which is the only case we additionally cut the last word into individual bits.
  \item The \emph{Long-Tail Case}: $T > a$ while $L > \Lr$.
  \end{itemize}

  The random numbers sampled on Lines~\ref{line:sample_start} to \ref{line:sample_end} serve as the random seed of the whole lemma, which does not change with variations in $m$. One can check that the random seed has at most $O(\log^2 N)$ bits, and $\numword, \numextbit$ are fully determined by $M$ and the seed, i.e., Condition~\ref{enum:no_extra_space} in \cref{lm:random_swap} holds.

  Based on this algorithm, we can transform the access and allocation/release operations in VM to those in WVM. As we know which words to swap (from $M$ and the seed), it is easy to transform an access to $m$ into several accesses to the WVM without additional time consumption (i.e., the time taken is linear in the number of transformed operations). Similarly, an allocation or release in the VM can be transformed to a series of necessary operations in the WVM without additional time.

  First, note that for any given operation in the VM, the number of transformed operations to the \emph{word tape} is at most $O(1)$, and specially, is exactly $1$ for a VM access.
  Therefore, it suffices to bound the expected number of operations to the \emph{bit tape}, which does not exceed the length of the bit tape $\numextbit$ times the probability of accessing the bit tape.
  We have $\numextbit = O(\min(L, w))$, because
  \begin{itemize}
  \item In the Few-Word Case, the number of bits equals $\numextbit = T \le 10L < w$ because the VM is word-dominant, and thus $T = O(\min(L, w))$.
  \item Otherwise, $\min(L, w) = \Theta(w)$, and the number of bits $\numextbit \le T + w \le 2w = O(\min(L, w))$.
  \end{itemize}
  Once the bit tape is accessed (except for the $O(1)$ inevitable operations when we allocate or release 1 bit on the VM), we pretend there are $\Theta(\min(L, w))$ bit-accesses (plus $\Theta(\min(L, w))$ allocations/releases if it was a VM allocation/release), as a clear upper bound of the real cost. Next, it only remains to bound the probability of accessing the bit tape.

  \paragraph*{Access operations.}
  The access to the bit tape happens only when we access the $\sk$-th original word, or the $\skplus$-th word in the Short-Tail Case. As $\sk$ is uniformly sampled from $[1, \Lk]$, the probability to access the bit tape is at most $2/\Lk \le 8/\Lk[k+1] \le 8/L$. Hence, the expected number of induced accesses to the bit tape is at most $O(\min(L, w)) \cdot (8/L) = O(\min(1, w/L)) = O(\min(1, w^2 / M))$ in expectation. Adding the 1 transformed access to the word tape implies \cref{lm:random_swap}~\ref{enum:access_of_WVM}.

  \paragraph*{Allocating/releasing a word.}
  By symmetry, we only consider the case of allocation. According to \cref{alg:random_swap}, the bit tape stores the substring of $m$ that starts from the first bit of the $\sk$-th word and has length $\numextbit$. Hence, the content of the bit tape changes only in the following two situations:
  \begin{itemize}
  \item The number of words $L$ crosses over the threshold $\Lr$ and switches from the Few-Word Case to the Short-Tail Case. In this situation, $\numextbit$ will increase by $w$. As $\Lr$ is uniformly sampled from $[w/20, w/11)$, the probability of this case is at most $O(1/w)$.
  \item The number of words $L$ crosses over the threshold $\Lk[k+1]$ and changes the index $k$, thus $\sk$ is replaced by $s_{k+1}$. As $\Lk[k+1]$ is uniformly sampled from $[2^{k+1}, 2^{k+2})$, the probability of this case is at most $1/2^{k+1} = O(1/L)$.
  \end{itemize}
  Combining them together, the probability of changing the bit tape is at most $O\bk{\max\bk{1/w, 1/L}}$, hence the expected number of operations in the bit tape is at most $O\bk{\max\bk{1/w, 1/L} \cdot \min\bk{L, w}} = O(1)$.

  \paragraph*{Allocating/releasing a bit.} In most of the cases during the allocation, $\numextbit = T$ or $T + w$ only increases by 1, so we only need to (1) allocate 1 bit on the bit tape, and (2) write into the newly allocated bit the swapped content (the $(T+1)$-th bit of some preceding word), which takes $O(1)$ bit-tape operations. (The contents in the original $\numextbit$ bits remain unchanged, so we do not need to access them.) Other additional bit-tape operations will only happen in the following situations:
  \begin{itemize}
  \item The number of words $L$ crosses over the threshold $\Lr$ and switches from the Few-Word Case to the Short-Tail Case. As argued above, its probability is at most $O(1/w)$.
  \item The number of words $L$ crosses over the threshold $\Lk[k+1]$ and changes the index $k$. As argued above, the probability is at most $O(1/L)$.
  \item The number of bits in the tail $T$ crosses over the threshold $a$ and switches from the Short-Tail Case to the Long-Tail Case, thus decreases $\numextbit$ by $w-1$. As $a$ is uniformly sampled from $[0, w-2]$, the probability of this case is at most $O(1/w)$.
  \end{itemize}
  Notice that when we allocate 1 bit upon $T = w - 1$, changing $T$ to 0, no additional action is required: The initial state cannot be the Few-Word Case, as there are at most $T \le 10L \le 10\Lr < w - 1$ bits in the tail; since $a \in [0, w - 2]$, this allocation turns the Long-Tail Case to the Short-Tail Case, thus $\numextbit$ only increases by one, which falls back to the ``common case'' mentioned above.
  
  Combining them together, the probability of additional bit-tape operations is at most $O\bk{\max\bk{1/w, 1/L}}$. Hence, the expected number of transformed operations is still at most $O\bk{\max\bk{1/w, 1/L} \cdot \min\bk{L, w}} = O(1)$.
\end{proof}

\section{Hash Permutation Families}
\label{sec:hashfunc}

\newcommand{\xl}{x_{\smallsub \textup{L}}}
\newcommand{\xr}{x_{\smallsub \textup{R}}}
\newcommand{\xli}[1][i]{x_{\smallsub \textup{L}, #1}}
\newcommand{\xri}[1][i]{x_{\smallsub \textup{R}, #1}}
\newcommand{\redh}{\hat{h}}

In this section, we prove \cref{lem:hashfamily_1,lem:hashfamily_2}. They assert existence of hash permutation families with low time and space usage as well as good concentration guarantees. The main technique used here is ``Feistel permutation'':

\begin{definition}[Feistel permutation]
  \label{def:feistel}
  For any $x \in \BK{0, 1}^L$ and $1 \le s < L$, we denote by $\xl, \xr$ the first $s$ bits and the last $L - s$ bits of $x$, respectively, written $x = \xl \concat \xr$. Fixing a function $f : \BK{0, 1}^{L - s} \to \BK{0, 1}^s$, its \emph{Feistel permutation} $\pi_f : \BK{0, 1}^L \to \BK{0, 1}^L$ is defined as
  \[
    \pi_f : x \mapsto \bk{\xl \oplus f(\xr)} \concat \xr.
  \]
\end{definition}

It is easy to verify that $\pi_f$ is indeed a permutation over $\BK{0, 1}^L$, and $\pi_f\bk{\pi_f(x)} = x$.

In order to construct hash permutation family $\hashfam$, we only need to construct a family $\funcfam$ of hash functions, which induces a family of Feistel permutations $\hashfam = \BK{\pi_f : f \in \funcfam}$. For \cref{lem:hashfamily_1}, we use the following hash function family.

\begin{claim}[\normalfont\cite{thorup2013simple}]
  \label{clm:hashfunc_1}
  For any fixed constant $c$, there is a hash function family $\funcfam$, in which every function
  \[
    f : \BK{0, 1}^{L - s} \to \BK{0, 1}^s
  \]
  can be evaluated in $O(c)$ time and can be stored in $O(2^{(L - s) / c} \cdot L)$ bits. Moreover, $\funcfam$ is $\Omega(2^{(L - s) / c^2})$-wise independent.
\end{claim}

Based on this hash function family, we prove \cref{lem:hashfamily_1}.

\HashFamilyOne*

\begin{proofof}{\cref{lem:hashfamily_1}}
  Applying \cref{clm:hashfunc_1} with $c = 2/\eps$ gives a hash function family $\funcfam$; let $\hashfam = \BK{\pi_f : f \in \funcfam}$ be the induced family of Feistel permutations. It is clear that any $h \in \hashfam$ can be evaluated in constant time (as well as its inverse $h^{-1} = h$), and can be stored within $O(2^{(L - s) \eps / 2} \cdot L) = O(2^{\eps L})$ bits.

  It remains to check \ref{enum:no_overflow_in_lemma}. For convenience, we use $\redh(x)$ to represent the first $s$ bits of $h(x)$. By the definition of Feistel permutations, we know $\redh(x) = \xl \oplus f(\xr)$ (assuming $x = \xl \concat \xr$). We make the following claim:

  \begin{claim}
    \label{clm:chernoff}
    Assume $n = \Omega(2^s \cdot s^4)$. We arbitrarily fix distinct strings $x_1, \ldots, x_n \in \BK{0, 1}^L$, $y \in \BK{0, 1}^s$, and sample $h \in \hashfam$ uniformly at random. Let random variable $X \defeq \bigabs{\redh^{-1}(y)}$ denote the number of $i \in [n]$ with $\redh(x_i) = y$, then
    \[
      \Pr_{h \in \hashfam} \Bk{\abs{X - \frac{n}{2^s}} \ge \frac{1}{5} \bk{\frac{n}{2^s}}^{2/3}} \le \frac{1}{4n^3}.
    \]
  \end{claim}

  We first assume the claim holds. By the union bound, the probability that there is an equivalent class whose number of elements lies outside $\Bk{\frac{n}{2^s} - \frac{1}{5} \bk{\frac{n}{2^s}}^{2/3}, \, \frac{n}{2^s} + \frac{1}{5} \bk{\frac{n}{2^s}}^{2/3}}$ is at most
  \[
    \Pr_{h \in \hashfam} \Bk{\exists y \in \BK{0, 1}^s : \bigabs{\redh^{-1}(y)} \notin \Bk{\frac{n}{2^s} - \frac{1}{5} \bk{\frac{n}{2^s}}^{2/3}, \, \frac{n}{2^s} + \frac{1}{5} \bk{\frac{n}{2^s}}^{2/3}}} \le \frac{1}{4n^3} \cdot 2^s \le \frac{1}{4n^2},
  \]
  which implies Property~\ref{enum:no_overflow_in_lemma}. Now it only remains to prove this claim.

  \bigskip

  Assume $x_i = \xli \concat \xri$ where $\xli$ is the first $s$ bits of $x_i$. We divide $x_1, \ldots, x_n$ into groups according to the right part $\xri$. Formally, we let $\BK{z_1, \ldots, z_m} = \BK{\xri[1], \ldots, \xri[n]}$ be the set of occurred $\xr$'s, and let $I_j \defeq \BK{i \in [n] : \xri = z_j}$.

  For $i_1 \ne i_2 \in I_j$, we know $\redh(x_{i_1}) \ne \redh(x_{i_2})$. It is because $\xri[i_1] = \xri[i_2] = z_j$, thus $\redh(x_{i_1}) = \xli[i_1] \oplus f(z_j) \ne \xli[i_2] \oplus f(z_j) = \redh(x_{i_2})$. Therefore, for every $j \in [m]$, there is at most one $i \in I_j$ with $\redh(x_i) = y$. Let $X_j \defeq \ind{\exists i \in I_j : \redh(x_i) = y}$. Then we have
  \begin{itemize}
  \item $X_j = \ind{\tall \exists i \in I_j : \xli \oplus f(z_j) = y} = \ind{\tall f(z_j) \in \BK{\xli \oplus y : i \in I_j}}$ is solely determined by $f(z_j)$. For a random $f \in \funcfam$, its expectation is $\E[X_j] = \frac{|I_j|}{2^s}$.
  \item $X \defeq \bigabs{\redh^{-1}(y)} = \bigabs{\bigBK{i \in [n] : \redh(x_i) = y}} = X_1 + \cdots + X_m$. Its expectation equals $\E[X] = \sum_{j=1}^m \frac{|I_j|}{2^s} = \frac{n}{2^s}$.
  \item Let $k \defeq 10 \ln n$, then any $k$-subset of $\BK{X_1, \ldots, X_m}$ is independent. This is because $\funcfam$ is $\Omega\bigbk{2^{(L-s) \eps^2 / 4}}$-wise independent, and
    \[
      k \defeq 10 \ln n \le (10 \ln 2) L = 2^{o(L)} \le o\bigbk{2^{(L - s) \eps^2 / 4}}.
    \]
  \end{itemize}
  We then apply the Chernoff bound with limited independence on $X = X_1 + \cdots + X_m$, stated below:

  \begin{theorem}[\protect{\cite[Theorem 5]{schmidt1995chernoff}}]
    \label{thm:chernoff}
    If $X$ is the sum of $k$-wise independent r.v.'s, each of which is confined to the interval $[0, 1]$, with $\mu = \E[X]$. For $\delta \le 1$, if $k \le \floor{\delta^2 \mu e^{-1/3}}$, then $\Pr\Bk{|X - \mu| \ge \delta \mu} \le e^{- \floor{k / 2}}$.
  \end{theorem}

  Now let $\delta = \frac{1}{5} \bk{\frac{n}{2^s}}^{-1/3}$ and apply \cref{thm:chernoff}. It is easy to check that
  \[
    \delta^2 \mu e^{-1/3} = \frac{1}{25 e^{1/3}} \bk{\frac{n}{2^s}}^{1/3} \ge \frac{1}{25 (2e)^{1/3}} \log^{4/3} n \gg k = 10 \ln n,
  \]
  where the first inequality holds since $\frac{n}{\log^4 n} \ge \frac{2^s \cdot s^4}{(s + 4 \log s)^4} \ge \frac{1}{2} \cdot 2^s$ (for sufficiently large $s$), which is equivalent to $\frac{n}{2^s} \ge \frac{1}{2} \cdot \log^4 n$. \cref{thm:chernoff} gives
  \[
    \Pr\Bk{\abs{X - \frac{n}{2^s}} \ge \frac{1}{5} \bk{\frac{n}{2^s}}^{2/3}} \le e^{- \floor{k / 2}} \le e^{- 5 \ln n + 1} \ll \frac{1}{4n^3}.
  \]
  This proves \cref{clm:chernoff} and further implies \cref{lem:hashfamily_1}.
\end{proofof}

\cref{lem:hashfamily_2} can be proved similarly using Feistel permutations.

\HashFamilyTwo*

\begin{proofof}{\cref{lem:hashfamily_2}}
  It is well-known that there exists a 2-independent hash function family $\funcfam$, in which every function $f : \BK{0, 1}^{L - s} \to \BK{0, 1}^s$ can be stored within $O(L)$ bits and evaluated in constant time (e.g., polynomial hashing over a finite field $\F_{2^L}$). For distinct inputs $x_1 \ne x_2 \in \BK{0, 1}^{L - s}$ and outputs $y_1, y_2 \in \BK{0, 1}^{s}$, $\Pr[f(x_1) = y_1 \land f(x_2) = y_2] = 2^{-2s}$.

  We then verify that $\funcfam$'s Feistel permutation family $\hashfam \defeq \BK{\pi_f : f \in \funcfam}$ satisfies the requirements. First, $h \in \hashfam$ (which equals its inverse) can be evaluated in $O(1)$ time, and can be stored using $O(L)$ bits of space. Next, let $h \in \hashfam$ be chosen uniformly at random, and denote by $\redh(x)$ the first $s$ bits of $h(x)$. For any two distinct inputs $x_1 \ne x_2 \in \BK{0, 1}^{L-s}$, the probability of $\redh(x_1) = \redh(x_2)$ is at most $2^{-s}$, because:
  \begin{itemize}
  \item Assume $x_1 = \xli[1] \concat \xri[1]$ and $x_2 = \xli[2] \concat \xri[2]$. If $\xri[1] = \xri[2]$, we always have $\redh(x_1) = \xli[1] \oplus f(\xri[1]) \ne \xli[2] \oplus f(\xri[2]) = \redh(x_2)$.
  \item Otherwise, $f(\xri[1])$ and $f(\xri[2])$ are independent r.v.'s following the uniform distribution over $\BK{0, 1}^s$. Therefore, $\Pr[\redh(x_1) = \redh(x_2)] = \Pr[f(\xri[1]) \oplus f(\xri[2]) = \xli[1] \oplus \xli[2]] = 2^{-s}$.
  \end{itemize}
  Combining these two cases tells that $\Pr[\redh(x_1) = \redh(x_2)] \le 2^{-s}$ for all $x_1 \ne x_2$. Taking the union bound over all pairs among $x_1, \ldots, x_n$ gives Property~\ref{enum:no_collision}.
\end{proofof}

\paragraph*{Generalizations for the top level.} In \cref{sec:dabtree,sec:multihash}, when hashing all $N$ keys into $N / \enkeyi[1]$ buckets, the input domain of the hash function might not be a power of two, in which case we use the following generalization of \cref{lem:hashfamily_1}. It just replaces the rightmost $L-s$ bits of the domain ($\BK{0, 1}^{L-s}$) with a more general set $[Q]$.

\begin{restatable}{corollary}{HashFamilyOneCor}
  \label{cor:hashfamily_1}
  Let $s, T$ be integers where $1 \le s \le (1 - \Omega(1))(s + \log T)$. There is a hash function family $\hashfam$ in which every member $h \in \hashfam$ is a permutation over $\BK{0, 1}^{s} \times [T]$, satisfying:
  \begin{enumerate}[label=\textup{(\alph*)}]
  \item For any $h \in \hashfam$ and any input $x \in \BK{0, 1}^s \times [T]$, $h(x)$ and $h^{-1}(x)$ can be evaluated in $O(1)$ time.
  \item It takes $O\bigbk{(2^s \cdot T)^{\eps}}$ bits to store an $h \in \hashfam$, where $\eps$ is a pre-given constant which does not depend on $s$ and $T$.
  \item\label{enum:no_overflow_in_lemma_cor} For $n \ge 2^{s} \cdot s^4$ different inputs $x_1, \ldots, x_n$, if we divide $h(x_1), \ldots, h(x_n)$ into equivalent classes according to the first $s$ bits (i.e., the $\BK{0, 1}^s$ part) of $h(x_i)$, then with probability $\ge 1 - \frac{1}{4n^2}$, the number of elements in any equivalent class is between
    $\Bk{ \frac{n}{2^{s}} - \frac{1}{5} \bk{\frac{n}{2^{s}}}^{2/3}, \, \frac{n}{2^s} + \frac{1}{5} \bk{\frac{n}{2^s}}^{2/3}}$.
  \end{enumerate}
\end{restatable}

The proof is almost the same as \cref{lem:hashfamily_1}. For every $x = (\xl, \xr) \in \BK{0, 1}^s \times [T]$, we treat $\xl \in \BK{0, 1}^s$ as the left part and $\xr \in [T]$ as the right part, and define \emph{Feistel permutations} similarly to \cref{def:feistel}: For any function $f : [T] \to \BK{0, 1}^s$, its Feistel permuation is $\pi_f : x \mapsto (\xl \oplus f(\xr), \, \xr)$. Then, we take a function family $\funcfam$ from $[T]$ to $\BK{0, 1}^s$ according to \cref{clm:hashfunc_1}, and apply the same argument as in the proof of \cref{lem:hashfamily_1} to prove the corollary.

\bigskip

A similar corollary of \cref{lem:hashfamily_2} holds as well:

\begin{restatable}{corollary}{HashFamilyTwoCor}
  \label{cor:hashfamily_2}
  Let $s, T$ be positive integers. There is a hash function family $\hashfam$ in which every member $h \in \hashfam$ is a permutation over $\BK{0, 1}^s \times [T]$, satisfying:
  \begin{enumerate}[label=\textup{(\alph*)}]
  \item For any $h \in \hashfam$ and any input $x \in \BK{0, 1}^s \times [T]$, $h(x)$ and $h^{-1}(x)$ can be evaluated in constant time.
  \item It takes $O(s + \log T)$ bits to store an $h \in \hashfam$.
  \item\label{enum:no_collision_cor} For any integer $n$ and $n$ fixed distinct inputs $x_1, \ldots, x_n$, with probability at least $1 - n^2 / 2^s$, the first $s$ bits (i.e., the $\midBK{0, 1}^s$ part) of $h(x_1), \ldots, h(x_n)$ are pairwise distinct.
  \end{enumerate}
\end{restatable}

\section{Hashing for Section \ref{sec:adaptertree}}
\label{sec:hashfunc_for_adapter}

The construction for hash functions and the strategy of rehashing for \cref{sec:adaptertree} is similar to that in \cref{sec:hashfunc_for_multilevel,sec:rehashing_for_multilevel}. We introduce them in this appendix. All notations here are the same as in \cref{sec:adaptertree} unless otherwise specified.

Recall that, at the top level of the data structure, we used two hash functions $h^{(1)} : [U] \to [N / \nodewidth{1}] \times [2^{\idleni[1]}] \times [\quotientuni]$ and $h^{(2)}_u : [2^{\idleni[1]}] \to [B^{8h-1}] \times [2^{\idleni[8h]}]$ to map every key $x$ into a leaf node. Here, $h^{(2)}_u$ can be different for different root nodes, and we store $N / \nodewidth{1}$ independent instances of it, while for $h^{(1)}$ we only have one instance. Both hash functions need to prevent \emph{overflows} and \emph{underflows} with high probability, while $h^{(1)}$ also needs to prevent \emph{ID collisions}, which means the same ID is assigned to two keys $x_1, x_2$ that are hashed to the same level-1 bucket. If either of these conditions fail, we resample the failed hash functions and reconstruct part of the data structure.

Similar to \cref{sec:hashfunc_for_multilevel}, $h^{(1)}$ is composed of two (bijective) subfunctions: $\permCh : [U] \to [N / \nodewidth{1}] \times [2^{\idleni[1]} \quotientuni]$ and $\permId[u] : [2^{\idleni[1]} \quotientuni] \to [2^{\idleni[1]}] \times [\quotientuni]$, where the former extracts the bucket index that the key is hashed to, and the latter extracts the key's ID within that bucket. The second subfunction $\permId[u]$ can be different for different level-1 buckets. $\permCh$ and $\permId[u]$ are sampled uniformly at random from the permutation families given by \cref{cor:hashfamily_1} and \cref{cor:hashfamily_2} respectively, restated below, with $(s, T, \eps) = \bigbk{\log (N / \nodewidth{1}), \, 2^{\idleni[1]} \quotientuni, \, \frac{1}{100(1+\alpha)}}$ and $(s, T) = (\idleni[1], \quotientuni)$ respectively.

\HashFamilyOneCor*

\HashFamilyTwoCor*

Property~\ref{enum:no_overflow_in_lemma_cor} in \cref{cor:hashfamily_1} tells that the probability of existing an overflowing or underflowing level-1 bucket does not exceed $\frac{1}{4N^2}$. Once overflows or underflows occur, we immediately pick different hash functions and reconstruct the whole dictionary, using at most $O(Nh)$ time. %

Property~\ref{enum:no_collision_cor} in \cref{cor:hashfamily_2} tells that the probability of ID collisions within any given level-1 bucket does not exceed $\nodewidth{1}^2 / 2^{\idleni[1]} \le \nodewidth{1}^{-2}$. Once ID collisions occur at level-1 bucket $u$, we resample $\permId[u]$ (but not $\permCh$) and rebuild the adapter tree below $u$, which takes $O(\nodewidth{1} h)$ time.

The construction of $h^{(2)}_u$ is simpler -- we directly sample $h^{(2)}_u$ according to \cref{cor:hashfamily_1} with $s = (8h-1) \log B$, $T = 2^{\idleni[8h]}$, and $\eps = 1/100$. The probability of overflows or underflows occurring is at most $\frac{1}{4\nodewidth{1}^2}$, and when they occur, we reconstruct the adapter tree below $u$ using $O(\nodewidth{1} h)$ time.

The additional space spent to store these hash functions is limited: we need to store one $\permCh$, using $O\bigbk{U^{1/100(1+\alpha)}} = O\bigbk{N^{1/100}}$ bits, and $N / \nodewidth{1}$ instances of $\permId[u]$ and $h^{(2)}_u$, using $O\bigbk{\frac{N}{\nodewidth{1}} \cdot (\idleni[1] + \log \quotientuni + 2^{\idleni[1] / 100})} = O\bigbk{\frac{N}{\nodewidth{1}} \cdot \nodewidth{1}^{1/25}} = o(R)$ bits. Moreover, through a similar argument to \cref{sec:rehashing_for_multilevel}, we can see that the expected amortized time complexity for rehashing is $O(1)$ per operation.

\paragraph{Achieving strong history-independence using fail modes.} The above approach completes our dictionary with $R = n / \log^{O(1)} n$ redundancy and $O(\log \log n)$ operational time. It is almost strongly history-independent except that we resample the hash functions when they fail -- the current hash functions carry information about the history. Removing the rehashing step will lead to a strongly history-independent dictionary with the same performance.

Our method is to fix the hash functions a priori, and when a hash function fails, we design a ``fail mode'' for the corresponding bucket, handling all operations slowly but still correctly. Since the probability of enabling the fail mode is extremely low, the expected time consumption of the fail mode is still $o(1)$. When the failure disappears after several subsequent updates, we reconstruct the memory state for the bucket and return to the ``normal mode''. Whether we enable the fail mode only depends on the current key set and the fixed (random) hash function, thus the entire data structure with fail modes is strongly history-independent.

Formally, let $u$ be a level-1 bucket, which is responsible for storing a set of $n(u) = \poly \log N$ elements from a universe $U_1 = 2^{\idleni[1]} \cdot \quotientuni$. If $\permId[u]$ incurs an ID collision when it tries to assign an ID to each key in bucket $u$, or $h_u^{(2)}$ incurs an overflow or underflow at a level-$8h$ (leaf) bucket, we say $u$ is in the \emph{fail mode}; otherwise, $u$ is in the \emph{normal mode}. \Cref{sec:dabtree} already shows how to store the set efficiently for the normal mode; in the fail mode, we present the following lemma.

\begin{lemma}[Fail Mode]
  \label{lem:fail_mode}
  Let $U'$ and $n$ be integers where $\log U' \le n$. There is a data structure storing $n$ elements from the universe $[U']$ using $\log \binom{U'}{n} + O(\log N)$ bits, while supporting all operations in $O(n^{100})$ time.
\end{lemma}

\begin{proof}[Proof Sketch]
  We list all $n$-element sets $S = \midBK{a_1, a_2, \ldots, a_n} \subseteq [U']$ ($a_1 < a_2 < \cdots < a_n$) in increasing lexicographic order of the sequences $(a_1, \ldots, a_n)$. For any given set $S$, assuming it occurs as the $k$-th set in the list, we define its encoding to be the binary representation of $k$. Taking into account the $O(\log N)$-bit overhead for storing $n$ itself, the total encoding length is $\log \binom{U'}{n} + O(\log N)$ bits as desired. It is not difficult to prove that both encoding $S$ into $k$ and decoding $S$ from $k$ takes $O(n^{100})$ time, hence all types of operations can be completed in $O(n^{100})$ time as well.
\end{proof}

We can assume the fail probability of any level-1 bucket is $O(n_1^{-200})$ instead of $O(n_1^{-2})$ as described above by adjusting the parameters in the analysis accordingly. For any given operation that accesses the level-1 bucket $u$, with probability at most $O(n_1^{-200})$, $u$ will be in the fail mode, which costs $O(n_1^{100})$ time to perform the operation due to \cref{lem:fail_mode}. The expected time overhead is thus $o(1)$.

We react similarly when the top-level hash function $\permCh$ fails (i.e., when any level-1 bucket overflows or underflows): We use the fail-mode structure to store the whole dictionary, which takes no more than $O(N^{100})$ time per operation. Because this failure only happens with probability $O(N^{-200})$, the expected overhead is negligible.

\end{document}